\documentclass[jair,twoside,11pt,theapa]{article}
\usepackage{jair_arxiv, theapa, rawfonts}

\jairheading{}{}{}{}{ }
\ShortHeadings{Price of Anarchy in Auctions}
{Roughgarden, Syrgkanis, Tardos}
\firstpageno{1}

\usepackage[suppress]{color-edits}
\usepackage{epsfig}
\addauthor{et}{green}
\addauthor{te}{blue}
\addauthor{vs}{red}
\addauthor{tr}{yellow}

\usepackage{definitions}

\begin{document}

\title{The Price of Anarchy in Auctions}

\author{\name Tim Roughgarden \email tim@cs.stanford.edu \\
       \addr Computer Science Department, Stanford University,\\
       Stanford, CA  USA
       \AND
       \name Vasilis Syrgkanis \email vasy@microsoft.com \\
       \addr Microsoft Research,
       1 Memorial Drive,\\
       Cambridge, MA USA
       \AND
       \name \'Eva Tardos \email eva@cs.cornell.edu \\
       \addr Computer Science Department, Cornell University,\\
       Ithaca, NY USA}


\maketitle

\begin{abstract}

This survey outlines a general and modular theory for proving
approximation guarantees for equilibria of auctions in complex
settings.  This theory complements traditional economic techniques,
which generally focus on exact and optimal solutions and are
accordingly limited to relatively stylized settings.

We highlight three user-friendly analytical tools: smoothness-type
inequalities, which \teedit{immediately} yield approximation guarantees for many auction
formats of interest in the special case of complete information and
deterministic strategies; extension theorems, which extend such
guarantees to randomized strategies\teedit{, no-regret learning outcomes,} and incomplete-information
settings; and
composition theorems, which extend such guarantees from simpler to
more complex auctions.  Combining these tools yields tight
worst-case
approximation guarantees for the equilibria of many widely-used
auction formats.

\end{abstract}

\section{Introduction}
\label{sec:introduction}

Many modern applications in computer science involve a number of
self-interested participants, with objectives different from each
other and from the application designer.
{\em Auctions} are a canonical genre of such applications, ranging
from the sale of antiques on eBay, to real-time and targeted
Internet advertising, to the allocation of licenses for the wireless
spectrum that constitutes much of the modern communication
infrastructure.

Auctions have been studied by economists for over a half-century.  But
over the last several years, work in computer science has offered a
fresh and relevant perspective.
This new theory is the subject of this survey and it
concerns {\em approximation guarantees} for the equilibrium performance
of auctions, also known as ``price of anarchy'' bounds~\cite{Koutsoupias1999}.
There are a number of fundamental models that appear impossible to
reason about without resorting to approximation.%
\footnote{Computer science
  has also brought other important ideas to the table in auction
  design, including an emphasis on reasonable computational complexity
  and robustness to informational assumptions.  These topics are
  outside the scope of this survey.}

For example, consider a seller with $m$ different items for sale.
There are $n$ bidders, and the seller does not know what the bidders
want.  What should the seller do?  One idea is to simply ask the
bidders what they want, meaning ask them to bid on each of the
possible subsets of items that they might get.  With a single item
($m=1$), this is a practical idea, and is basically what happens in an
eBay auction.  In general, however, this idea requires soliciting
$2^m$ bids from each bidder (one per subset of the items), which is a
non-starter unless $m$ is tiny.

So how are multiple items auctioned off in practice?  One of the
most common methods is to sell each item separately.  That is, each
bidder submits one bid for each item ($m$ bids in all), and each item
is awarded to the highest bidder on that item (e.g., for a price equal
to the highest or second-highest bid).  How good is this simple
method --- do bidders bid in a way that the items are allocated to those
who want them the most?

It has been known for many decades that simple auction formats do not
generally result in the most efficient allocation of the items for
sale, but conventional wisdom in economics states that the allocation
should be ``pretty good'' provided bidders' preferences over items are
``sufficiently nice'' (see e.g.~\citeA{milgrombook}).
Traditional economic tools appear inadequate
for translating this empirical rule of thumb into a rigorous
performance guarantee, for two reasons: (1) work in economics has
focused on exact and optimal solutions, and for the most
part has not considered approximation guarantees; (2) economic systems
are traditionally studied by solving for and then analyzing the
equilibria, while the equilibria of multi-item auctions are far too
complex to characterize.

This survey outlines a fairly general and modular theory for proving
rigorous performance guarantees for equilibria of auctions in complex settings.  A
representative consequence of the theory is: if multiple items are
sold separately using first-price auctions and the willingness to pay
of each bidder is a submodular function of the items that she receives
(i.e., preferences are ``sufficiently nice''), then every equilibrium
of the auction achieves social welfare
at least 63\% of the maximum possible (i.e., is
``pretty good'').  Despite the complexity of the welfare-maximization
problem and of the equilibria in such multi-item auctions,
this
guarantee follows from a user-friendly three-step recipe.
Roughly, the first step is to consider only the very special case of a
single item, with every bidder knowing the willingness to pay of every
other bidder, and pure (i.e., deterministic) equilibria.
The task in this special case is to use elementary arguments to
translate the defining conditions of an equilibrium into a particular
type of approximate welfare guarantee (a ``smoothness-type
inequality'').  The second step is to apply an ``extension theorem,''
which extends the approximate welfare guarantee to single-item
auctions in which there is uncertainty (i.e., randomness) both in what
bidders are willing to pay and in what bids they submit\teedit{, or 
to outcomes reached by players using no-regret learning, which results in a 
form of correlation in this randomness}.  The third
step is to apply a ``composition theorem,'' which extends the
approximate welfare guarantee
to equilibria
of simultaneous 
auctions with any number of items.

\paragraph{Organization.}
Section~\ref{sec:fpa} explores a simple but non-trivial example: the
worst-case inefficiency of equilibria in first-price single-item
auctions.  In addition to introducing Bayes-Nash equilibrium
analysis in a concrete and understandable setting,
the efficiency analysis in this section already introduces the
essence of many of the
key ideas of the general framework.  Section~\ref{sec:sim-fpa}
considers a more complex example, the sale of multiple items via
simultaneous single-item auctions.  After absorbing the analyses of
these two examples, the general theory for
approximate efficiency guarantees for
smooth auctions via ``extension theorems'' described in Section~\ref{sec:model} follows
naturally.  Section~\ref{sec:no-regret} explains why guarantees for
smooth auctions apply even when players fail to converge to an
equilibrium, provided each achieves a ``no-regret'' property over
repeated plays of the auction.  Section~\ref{sec:composability}
shows that the extension from single-item to simultaneous single-item
auctions (Sections~\ref{sec:fpa} and~\ref{sec:sim-fpa}) is a general
phenomenon, by proving a ``composition theorem'' for smooth
auctions.  Section~\ref{sec:lower_bounds} considers the limitations of
simple auctions, and explains why they cannot enjoy good
price-of-anarchy bounds with general bidder preferences.
Section~\ref{sec:other-topics} offers pointers to the literature on
related topics not covered in this survey, and Section~\ref{sec:open}
concludes with a number of open research directions.



\section{A Simple Example: First-Price Single-Item Auctions}
\label{sec:fpa}

\subsection{First-Price Auctions and Bayes-Nash Equilibria}


We begin our analysis of the equilibria of auctions with a simple but
fundamental non-truthful auction, the first-price single-item auction.
Consider a single item being auctioned off to one of $n$ bidders (also
called ``players''). Each
bidder $i$ has some value $\vali$ for winning the item --- the maximum
``willingness to pay'' of the bidder --- which is private information
known only to her.  If she wins the item and is asked to pay a price
$p_i$, then her payoff is $\vali-p_i$.
We refer to such payoffs by saying that players have quasi-linear
preferences with respect to money.\footnote{Non-quasi-linear utility
  functions, such as those incorporating risk aversion or budgets, are
  also interesting.  See also 
  Sections~\ref{sec:other-topics} and~\ref{sec:open}.}

In a sealed-bid first-price auction, every player~$i$ simultaneously
submits a bid $\bidi$ to the auctioneer. The player with the highest
bid wins the item and is asked to pay her bid. Ties are broken
arbitrarily.


Bidding in a first-price auction is tricky.\footnote{By contrast, in a
  {\em second-price auction}, where the winning bidder pays the value
  of the second-highest bid, it is a weakly dominant strategy to bid
  to bid one's true value.}
Certainly no player will bid her true valuation, as this would
guarantee zero payoff.  Instead, bidders ``shade'' their bids, meaning
bid less that their values.  By how much should a bidder shade her
bid?  The answer depends on the amount of competition she faces and on
how other players behave.
How can a bidder reason about what others will do when their
valuations are unknown to her?  The standard approach to modeling this
issue is via (Bayesian) {\em games of incomplete information}.
We assume that the valuation $\vali$ of each player $i$ is drawn
independently from
some distribution $\F_i$, and that these distributions are common
knowledge to all of the players.
Intuitively, these distributions correspond to the common
beliefs that players have about everyone's valuations.


In this incomplete-information model, a {\em strategy} of a player
is a function $\sdi$ that maps a value $\vali$ in the support of
$\F_i$ to a bid $\sdi(\vali)$.  The semantics are: ``when my valuation
is $\vali$, I will bid $\sdi(\vali)$.''
The central equilibrium concept in Bayesian games is the
\emph{Bayes-Nash equilibrium}.  By definition, a profile of strategies
constitutes a Bayes-Nash equilibrium if for every player~$i$ and every
valuation~$\vali$ that the player might have, the player chooses a
bid~$\sdi(\vali)$ that maximizes her conditional expected utility.
The expectation is over the valuations of other players, conditioned
on bidder~$i$'s valuation being~$\vali$.

\subsection{Symmetric Valuation Distributions}\label{ss:sym}


What do Bayes-Nash equilibria look like in a single-item first-price
auction?  To get a feel for this question, we begin with a simple
example, of two bidders with valuations drawn independently and
identically from the uniform distribution on $[0,1]$.

\begin{example}{Two bidders with uniform $[0,1]$ valuations}\label{ex:simple}
%
Let's ``guess and check'' a Bayes-Nash equilibrium for this example.
First, since the setting is symmetric in the two bidders, it is
natural to guess that Bayes-Nash equilibria are also symmetric, meaning
that the two players use the same strategy $\sd(\val)$.
Let's also guess that the function $\sd(\cdot)$ is strictly
increasing, continuous, and differentiable. Under these assumptions,
the highest bidder is the bidder with the largest valuation.  By
symmetry, the probability that a bidder with valuation~$\val$ wins is
$\F(\val) = \val$.  To check the Bayes-Nash equilibrium conditions,
fix a player and condition on her valuation being~$\val$.  We need to
solve for the bid that maximizes the expected utility of the bidder.
\etedit{She could pretend to have value $z$ and bid $\sd(z)$ for $z \in [0,1]$.
Her expected utility for such a bid} is
\begin{equation*}
g(z) =\underbrace{(\val- \sd(z))}_{\text{utility of win}}\cdot
\underbrace{\F(z)}_{\text{prob.\ of win}}  = (\val-\sd(z))\cdot z.
\end{equation*}
To force the condition that the optimal bid of the form~$\sd(z)$ is
$\sd(\val)$, as prescribed by the Bayes-Nash equilibrium conditions,
we differentiate~$g$ with respect to~$z$ and set~$s(\cdot)$ to force a
zero derivative at~$\val$.  This
yields the condition
\begin{equation*}
0 = g(z)'|_{z=v}=v-(s(z)z)'|_{z=v} = \val-(\val\cdot \sd(\val))'  \Leftrightarrow \val \cdot \sd(\val) = \frac{\val^2}{2}+\text{constant}
\end{equation*}
on the function~$\sd(\cdot)$.
Setting $\sd(0)=0$, we obtain the solution $\sd(\val) =
\frac{\val}{2}$. This solution does indeed satisfy our initial
assumptions of differentiability and monotonicity.  It is also easy to
check that it satisfies the Bayes-Nash equilibrium conditions for all
bids, and not just for bids of the form~$\sd(z)$ for some~$z$.

This Bayes-Nash equilibrium turns out to be the unique equilibrium in
this example.  This equilibrium is fully efficient, in that the item
is always allocated to the bidder with the higher valuation.
\end{example}

The argument in Example~\ref{ex:simple} generalizes to arbitrary
settings in which players' valuations are drawn independently and
identically from a distribution $\F$.  Bayes-Nash equilibria also
continue to be unique in this case~\cite{Chawla2013}.
We refer to such settings as
\emph{symmetric first-price auctions}.  For example, with $n$ bidders
with valuations drawn from the uniform distribution on $[0,1]$, every
player uses the strategy $\sd(\val)=\frac{n-1}{n}\val$ in the
Bayes-Nash equilibrium.  Thus, as competition increases, bidders shade
their bids less at equilibrium.

\trdelete{
of what is called the \emph{revenue equivalence} principle
\cite{Myerson1981}. Revenue equivalence states that if two auctions
always have the same winner than the player's expected payment must be
the same also. In a symmetric equilibrium of a single-item auction
with monotone bidding function, the player with highest value
wins. This is the same as the outcome of the second price auction, and
hence  by revenue equivalence, the bid of a player with value $\val$
in a symmetric first-price auction equals \etedit{her} expected
payment in a second price auction, \etedit{the} highest value among
her opponents, conditional each one of the opponents having a
valuation smaller than $v$.
}


Summarizing, there are two key take-aways about symmetric first-price
auctions.
\begin{enumerate}

\item Bayes-Nash equilibria are relatively well understood.

\item Bayes-Nash equilibria are fully efficient, with the item always
  allocated to the bidder with the highest valuation.

\end{enumerate}

\subsection{Asymmetric Valuation Distributions}


When
there is information that distinguishes different bidders,
for example the market shares of different companies, the symmetry
assumption of Section~\ref{ss:sym} is no longer appropriate.
Can we extend the results of that section to {\em
  asymmetric} first-price single-item auctions
\cite{Maskin2000},
where bidders valuations' are drawn from different distributions?


This question has been extensively studied; see Section 4.3 of
\citeA{Krishna2002}.
Because solving for a Bayes-Nash equilibrium in the asymmetric case is
a daunting task and generally admits no closed-form solution,
most papers in the area have considered only the case
of two bidders and specific parametric distributions
\cite{Vickrey1961,Kaplan2012}.
Already with two bidders with valuations drawn uniformly from $[0,1]$
and $[0,2]$, things get complicated.

\begin{example}{Two bidders with uniform $[0,1]$ and uniform $[0,2]$
    distributions \cite{Vickrey1961}} \label{ex:vickrey}
One can verify that the following bidding functions constitute an equilibrium in this example (see also \citeA{Krishna2002}):
\begin{align*}
\sd_1(\val_1) =~& \frac{4}{3 \val_1} \left(1-\sqrt{1-\frac{3 \val_1^2}{4}}\right)\\
\sd_2(\val_2) =~& \frac{4}{3 \val_2} \left(\sqrt{1+\frac{3 \val_2^2}{4}}-1\right).
\end{align*}
%
Both bidders bid in the range $[0,\frac{2}{3}]$, with the
weaker bidder $1$ bidding more aggressively than the stronger bidder
$2$ (i.e., $\sd_1(\val) > \sd_2(\val)$ for $\val \in [0,1]$).
For intuition, recall from Example~\ref{ex:simple} that if the bidders
had uniformly and identically distributed valuations, then at
equilibrium both bid half their value.  Recall also that the
equilibrium bid of a player increases with the amount of competition
she faces.
In this example, from the first bidder's perspective, the other bidder
represents stiffer competition than an identically distributed bidder,
so she bids more aggressively than in the symmetric case.
The opposite reasoning applies to the second bidder, who bids
less aggressively than in the symmetric case.
For this reason, the Bayes-Nash equilibrium is not fully efficient ---
there are valuation profiles in which the bidder with the lower
valuation is the higher bidder and hence the winner.
\end{example}



Summarizing, even the simplest asymmetric first-price auctions are
less well-behaved than symmetric ones, in two senses.
\begin{enumerate}

\item
Solving for a Bayes-Nash equilibrium requires finding a solution to a
system of partial differential equations, which in most cases has no
closed-form solution.

\item Bayes-Nash equilibria are generally inefficient, in that
the bidder with the highest valuation is not always the winner.

\end{enumerate}


How inefficient can the Bayes-Nash equilibria of auctions be?
The goal of this survey is to explain a number of general tools that
have proved useful for answering this question, along with several
representative applications.
Specifically, for every auction format that we consider,
we aim to show that the {\em price of anarchy}
--- the smallest ratio between the expected welfare of a Bayes-Nash
equilibrium and the expected welfare of a welfare-maximizing
allocation --- is at least some constant, independent of the
parameters of the auctions (the number of bidders, the valuation
distributions, etc.).

\subsection{The Price of Anarchy of First-Price Auctions}


The previous section demonstrates the futility of trying to
characterize the Bayes-Nash equilibria of asymmetric first-price
auctions in order to bound their inefficiency.
Instead, our analysis will rely only on the fact that, in
a Bayes-Nash equilibrium, every player is best responding to her
opponents' strategies.
We then show that every strategy profile that satisfies this best
response property is approximately efficient.
We prove that, in every (asymmetric) first-price auction, every
Bayes-Nash equilibrium has expected welfare at
least $\approx 63\%$ of the maximum possible.
Thus non-trivial efficiency equilibrium guarantees do not require a
detailed understanding of the structure of equilibria.


We use the following notation.  For a bid profile
$\bids=(\bid_1,\ldots,\bid_n)$, $\xsi(\bids)$ denotes whether or not
bidder~$i$ is the winner (1~or~0, respectively).
We denote by
$p(\bids)=\max_{i\in \{1,\ldots,n\}}\bid_i$ the selling price, which
is the highest bid.
We use $u_i(\bids;\vali)$ to denote the utility of player~$i$ when
her valuation is $\vali$ and the bid profile is $\bids$.  Because the
auction is first price, we can write
\begin{equation}\label{eq:qlsi}
u_i(\bids;\vali) = \left(\vali - \bidi\right)\cdot \xsi(\bids).
\end{equation}


Consider a strategy profile $\sds= (\sd_1,\ldots,\sd_n)$,
where each strategy $s_i$ is a function from the player's valuation to her
bid. We use $\sds(\vals)$ to denote the strategy vector resulting from
the vector of valuations $\vals$.
For a vector $\mathbf{w}$, we use $\mathbf{w}_{-i}$ to denote the
vector $\mathbf{w}$ with the $i$th component removed.  For example,
$\sdsmi(\valsmi)$ is the vector of bids of the players other than~$i$
when the valuations of these players are $\valsmi$.
With this notation, a strategy profile $\sds= (\sd_1,\ldots,\sd_n)$ is
a  Bayes-Nash equilibrium if and only if
\begin{equation}
 \E_{\valsmi}\left[u_i(\sds(\vals); \vali)~|~ \vali\right] \geq \E_{\valsmi}\left[u_i(\bid_i',\sdsmi(\valsmi); \vali)~|~\vali\right]
\end{equation}
for every player~$i$, every possible valuation~$\vali$ of the player,
and every possible deviating bid~$\bid_i'$.
The expectations are over the valuations of the players other
than~$i$,
according to the assumed prior distribution $\F$.

The {\em social welfare} of a bid profile $\bids$ when the valuation
profile is $\vals=(\val_1,\ldots,\val_n)$ is
\begin{equation}\label{eq:sw}
SW(\bids;\vals)=\sum_{i=1}^{n} \vali\cdot \xsi(\bids).
\end{equation}
\teedit{Note that the social welfare is the sum of the utilities of the
players plus the revenue of the auctioneer.} The maximum-possible social welfare in a single-item auction is
\begin{equation}\label{eq:opt}
\opt(\vals) = \sum_{i=1}^n \val_i \cdot \xsi^*(\vals),
\end{equation}
where $\xsi^*(\vals)$ is the indicator variable for whether or not
player $i$ is the player with the highest valuation
(with ties broken arbitrarily).


The {\em price of anarchy} of an auction, with valuation
distribution~$\F$, is the smallest value of the
ratio
\[
\frac{\E_{\vals}\left[SW(\sds(\vals); \vals)\right]}
{\E_{\vals}\left[\opt(\vals)\right]},
\]
ranging over all Bayes-Nash equilibria $\sds$ of the auction.
Thus inefficiency is measured by the extent to which the price of
anarchy is smaller than~1.
The price of anarchy of an auction format is then the worst-case
(i.e., smallest) price of anarchy of the auction in any setting,
ranging over all choices~$n$ for the number of players and all
valuation distributions~$\F$.


\begin{theorem}[Price of Anarchy of First-Price Single-Item Auctions]\label{thm:single-item-fpa}
The price of anarchy of the first-price single-item auction format is
at least $1-\tfrac{1}{e} \approx 0.63$.
\end{theorem}

\begin{proof}
We prove a bound of $\tfrac{1}{2}$; the bound of $1-\tfrac{1}{e}$
follows from an optimized version of the following argument
(see \citeA{Syrgkanis2013}).

Let $\sds$ be a Bayes-Nash equilibrium.
By definition, every player~$i$ chooses a strategy that maximizes her
expected utility, given her valuation $\vali$, the \etdelete{conditional}
distribution on $\valsmi$, and the strategies $\sdsmi(\valsmi)$ used
by the other players.  In particular, if a bidder~$i$ deviates from
the bid $\sdi(\vali)$ that she uses in the equilibrium, to bidding half
her value ($b_i^*=\frac{v_i}{2}$), then her expected utility can only
go down.\footnote{To obtain the
$1-\tfrac{1}{e}$ bound, one needs to consider a randomized bid with
support $\left[0,(1-\tfrac{1}{e})\vali\right]$ in
  place of the deterministic bid $\vali/2$.}
The choice of such hypothetical deviations~$b^*_i$ will be an
important theme throughout this survey.

It is simple to bound from below the utility of a bidder that bids
half her value ($b_i^*=\frac{v_i}{2}$).
For every bid profile $\bids$, we have
\begin{equation}\label{eqn:half-value-lower}
u_i(\bidi^*,\bidsmi;\vali) \geq \frac{1}{2}\vali - p(\bids),
\end{equation}
since the bidder either wins and obtains
utility $\vali-\bid_i^* =
\tfrac{1}{2}\vali \ge \tfrac{1}{2}\vali - p(\bids)$
or loses (in which case $\tfrac{1}{2}\vali < p(\bids)$) and obtains
utility~$0 \ge \tfrac{1}{2}\vali-p(\bids)$.
Since the bid $b_i^*=\frac{v_i}{2}$ guarantees non-negative utility,
we can also write
\begin{align*}
u_i(\bidi^*,\bidsmi;\vali)\geq \left(\frac{1}{2}\cdot \vali - p(\bids)\right)\cdot \xsi^*(\vals).
\end{align*}
Summing this inequality over all bidders~$i$, we obtain
\begin{align}\label{eqn:example-smoothness}
\sum_{i=1}^{n} u_i(\bidi^*,\bidsmi;\vali) \geq \sum_{i=1}^{n}\left(\frac{1}{2}\cdot \vali - p(\bids)\right)\cdot \xsi^*(\vals) = \frac{1}{2}\opt(\vals) - p(\bids)
\end{align}
for every valuation profile $\vals$ and bid profile $\bids$,
where $\xsi^*$ is defined as in~\eqref{eq:opt}.

We now invoke the hypothesis that $\sds$ is a Bayes-Nash equilibrium:
for every player~$i$ with valuation $\vali$,
\begin{align}\label{eq:br}
\E_{\valsmi}\left[u_i(\sds(\vals);\vali)\right] \geq \E_{\valsmi}\left[u_i(\bidi^*,\sdsmi(\valsmi);\vali)\right].
\end{align}
Taking expectations over $\vali$, summing up the $n$ inequalities of
the form~\eqref{eq:br},
and combining with inequality~\eqref{eqn:example-smoothness}, we obtain
\begin{equation*}
\sum_{i=1}^{n} \E_{\vals}\left[u_i(\sds(\vals);\vali)\right]  \geq  \sum_{i=1}^{n} \E_{\vals} \left[u_i(\bidi^*,\sdsmi(\valsmi);\vali)\right]  \geq \E_{\vals}\left[\frac{1}{2}\opt(\vals) - p(\sds(\vals)) \right].
\end{equation*}

Last, observe that by the quasi-linear form of bidders' utilities, for
every bid profile $\bids$ and valuation profile $\vals$ we have
\begin{equation*}
\sum_{i=1}^{n} u_i(\bids;\vali) = SW(\bids;\vals) - p(\bids).
\end{equation*}
Combining the last inequality and equation yields
\begin{equation*}
\E_{\vals}\left[SW(\sds(\vals);\vals)\right] =\sum_{i=1}^{n}\E_{\vals}[u_i(\sds(\vals);\vali)] + \E_{\vals}\left[p(\sds(\vals))\right]\geq \frac{1}{2}\E_\vals\left[\opt(\vals)\right],
\end{equation*}
which concludes the proof. 
%
\end{proof}

\begin{remark}
It is natural to ask if the bound of $1-1/e$ in
Theorem~\ref{thm:single-item-fpa} is the best possible.
All that is currently known is that no bound better that $.87$ is
possible \cite{Hartline2014}.
Determining the precise worst-case price of anarchy of asymmetric
first-price auctions
is an interesting open question (Section~\ref{sec:open}).
\end{remark}

Interestingly, the proof of Theorem~\ref{thm:single-item-fpa} never
used the assumption that bidders' valuations are independent.
We therefore have an even stronger guarantee.

\begin{theorem}\label{thm:single-item-corr-fpa} Every Bayes-Nash
  equilibrium of a first-price auction with correlated valuation
  distributions has expected social welfare at least $1-1/e$ times
  that of the expected optimal welfare.
\end{theorem}

\begin{remark}
For correlated valuation distributions, this bound of $1-1/e$ 
is tight (see \cite{SyrgkanisThesis} for an explicit example).
\end{remark}


\section{A More Complex Example: Simultaneous Single-Item First-Price Auctions}
\label{sec:sim-fpa}

\subsection{Multi-Item Auctions}

Before moving to a general auction setting,
we explore a more complex example and introduce the idea of smooth
auctions.
In this section, we consider an auction for multiple heterogeneous
items and quantify the equilibrium inefficiency of a simple
decentralized auction format.

Consider a set of $m$ items being auctioned off to a set of $n$
bidders. Each bidder $i\in [n]$ has a value $v_{ij}$ for each item
$j\in [m]$ and only wants one item. If she happens to win a set of
multiple items $S$, then her valuation for the set is the
highest-valued item in the set, i.e., $v_i(S) = \max_{j\in S}
v_{ij}$.
Such bidders are often called {\em unit-demand} bidders.
This
setting has a long and distinguished history in the economics
literature. It is a generalization of the \emph{assignment model}
analyzed by \citeA{Shapley1971}, where
the notion of \emph{core outcomes} was originally introduced,
and it is also the setting considered by
\citeA{Demange1986}, who introduced the first ascending
auctions for multi-item settings that implement welfare-optimal
outcomes. 
From an algorithmic point of view, the welfare optimization problem
in this setting is
simply the maximum weighted bipartite \etedit{matching} problem
\cite{Cook1997}, which also has a long and distinguished history in
combinatorial optimization.

This section analyzes the efficiency of the simple
auction format
that sells each item $j$ simultaneously and
independently using a single-item first-price
auction (Section~\ref{sec:fpa}).
Each player $i$
submits a bid $b_{ij}$ for each item $j$. Each item $j$ is
awarded to the highest bidder for the item and each bidder is asked to
pay her bid for each item that she won. The utility of a player is her
value for the set of items she won minus her total payment. Thus, if we
denote by $S_i(\bids)$ the set of items allocated to player $i$ under a
bid profile $\bids$, then:
\begin{equation}\label{eq:ql-sim}
u_i(\bids; v_i) = v_i(S_i(\bids)) - \sum_{j\in S_i(\bids)} b_{ij} =  \max_{j\in S_i(\bids)} v_{ij} - \sum_{j\in S_i(\bids)} b_{ij}
\end{equation}

The analysis of this auction game with unit-demand bidders dates back
to the early work of
\citeA{Engelbrecht1979} --- when auction theory was still in its
infant stages --- who analyzed the very special case of $n$ bidders
and $n$ items, with all players having a valuation of $1$ for each
item and constrained to bid on only one item.
The optimal welfare is clearly~$n$.
Interestingly, they show that there exists a symmetric mixed
Nash equilibrium whose expected welfare approaches $(1-1/e)n$
as $n \rightarrow \infty$.
In other words,
they showed that the price of anarchy of this auction game can be as
bad as $1-1/e$. 
This result can be viewed as the first price-of-anarchy
bound in the realm of simple auctions. Surprisingly, the
techniques that we describe in this survey will show that the \tedelete{price of
anarchy of any instance of this auction game is at least $1-1/e$.
We conclude that the} example proposed by \citeA{Engelbrecht1979}
exhibits the worst-possible inefficiency, over all numbers of bidders
and items, all choices of bidders' valuation distributions, and all
Bayes-Nash equilibria.

Two decades later, \citeA{Bikhchandani1999} analyzed a generalization
of the game where players can have arbitrarily complex valuations over
the items and without the restriction of bidding on only one item.
He focused on the special case of complete information, where all players'
valuations are common knowledge, and of pure-strategy Nash equilibria,
where each bidder deterministically chooses a single strategy.
He showed that complete-information pure-strategy Nash equilibria are
fully efficient.
\trdelete{
, and guaranteed to exist if and only if
strategies induces a welfare-optimal allocation. He
also showed that a Nash equilibrium with deterministic strategies
exists if and only if the player valuations are such that a
\emph{Walrasian equilibrium} exists in the \etedit{corresponding}
market \cite{??}.}
We are interested in the case of incomplete information and general
Bayes-Nash equilibria; as we know from the previous section, these are
not always fully efficient, even in the case of a single item.

The goal of this section is to prove an approximate efficiency
guarantee for simultaneous first-price auctions with arbitrary
independent distributions over unit-demand valuations.
\begin{theorem}\label{thm:sim-fpa-bne}
Every Bayes-Nash equilibrium of the simultaneous first-price auction
game with unit-demand bidders \etedit{and independent valuations}
achieves expected social welfare at least $1-1/e$ times the expected
optimal welfare.
\end{theorem}
We start with the simple case where players' valuations are
common knowledge. Unlike \citeA{Bikhchandani1999}, we analyze the
inefficiency of mixed-strategy Nash equilibria (where players can
randomize).
We will then see how our conclusions
for mixed Nash equilibria of the complete information setting also
extend to the incomplete information setting. In the remainder of the
survey we will see how this type of extension is not specific to the
game studied in this section, but applies more generally to many
auction environments.

\subsection{Warm-up: Complete Information}\label{ss:complete}

We begin with the case studied by \citeA{Bikhchandani1999}, of
pure-strategy complete-information Nash equilibria.
Fix a profile $\vals$ of unit-demand valuations.  Fix a
welfare-maximizing allocation, where without loss each bidder receives at
most one item, and let $j^*(i)$ denote the item awarded to bidder~$i$
in the allocation (if any).

At a pure Nash equilibrium,
each player $i$ submits a bid vector $b_{i}=(b_{ij})_{j\in [m]}$ such that, for all \etedit{vectors} $b_i'$:
\begin{equation}\label{eq:no-regret}
u_i(\bids; v_i) \geq u_i(b_i',\bidsmi; v_i).
\end{equation}
Let $p_j(\bids)$ denote the price at which item $j$ is sold.
Since a player $i$ does not gain from deviating to any other
strategy, she does not gain by bidding infinitesimally above the
current price on the item $j^*(i)$ and zero on all other items.
Denote this deviating bid vector by
$b_i^*$. By bidding $b_i^*$, player $i$ definitely wins his optimal
item and pays $p_{j^*(i)}(\bids)$, deriving utility
\begin{equation*}
u_i(b_i^*,\bidsmi;v_i) = v_{ij^*(i)} - p_{j^*(i)}(\bids)
\end{equation*}
Summing over the bidders~$i$, we conclude that there exist deviations
$b_1^*,\ldots,b^*_n$ such that
\begin{equation}\label{eqn:pne-smoothness}
\sum_{i\in [n]} u_i(b_i^*,\bidsmi;v_i)
=
\opt(\vals) - \sum_{j\in [m]}p_{j}(\bids).
\end{equation}
Using the fact that the sum of the Nash equilibrium utilities is at
least the sum of these deviating utilities
(applying~\eqref{eq:no-regret} with $b_i'=b_i^*$)
and that the total utility equals the welfare minus the revenue,
we can easily derive the following theorem.

\begin{theorem}[\citeA{Bikhchandani1999}]\label{thm:sim-fpa-pne}
Every complete-information pure Nash equilibrium of the simultaneous
first-price auction game with unit-demand players
achieves the maximum-possible welfare.
\end{theorem}

The above result may seem surprising --- every equilibrium outcome
corresponds to an optimal matching, the solution of a non-trivial
combinatorial optimization problem.  Thus decentralized optimization
by competing bidders yields a globally optimal solution.
But how robust is this result?

\trdelete{
 What if each player is not using a deterministic bid vector, but
 rather a randomized bid? A Nash equilibrium with deterministic
 strategies is not always guaranteed to exist in a finite normal form
 game, but a mixed Nash equilibrium always exists. Hence, having a
 result for a solution concept that always exists seems important.
}

Unfortunately, the analysis in the proof of Theorem
\ref{thm:sim-fpa-pne} breaks down when we try to apply it to a mixed
Nash equilibrium, even in the case of complete information.
Specifically, if the bid profile
is random, then the price $p_j(\bids)$ of an item is a random
variable.
A player is not in a position to deviate to bidding $p_j(\bids)$ on an
item, since she does not know the realization of $p_j(\bids)$ at the
time of bidding.
This issue cannot be mitigated by a different analysis,
as there exist inefficient mixed Nash equilibria of
the complete-information simultaneous first-price auction
game.\footnote{As mentioned above, this was first observed by
\citeA{Engelbrecht1979}. 
The price of anarchy remains at most $1-1/e$
even if players can bid on multiple items, if players can have
valuations slightly more general than unit-demand
\cite{Christodoulou2013}.}
\trdelete{
have slightly more general valuations called multi-unit-demand
or OXS
\footnote{Specifically, the valuations considered have the
  following form: each bidder partitions the items into groups
  $C_1,\ldots,C_t$. She wants only one item from each partition and
  has value $1$ for getting any item, while she is additive across
  partitions, i.e.,
  $v_i(S) = |\{r\in [t]: \exists j\in C_r\cap S\}|$.}
}

\begin{example}{Inefficiency of mixed Nash equilibria}\label{ex:badmne}
  Consider the case of two bidders and two items. Each bidder has a
  value of $1$ for any of the items and bidders are restricted to bid
  on at most one item. It is relatively easy to show that the
  following is a symmetric mixed Nash equilibrium of the game: each
  bidder picks an item uniformly at random and then submits a bid $x$
  on that item drawn from a distribution with cumulative density
  function
\begin{equation*}
F(x) =\frac{x}{1-x}
\end{equation*}
and support $\left[0,\frac{1}{2}\right]$. It is easy to check that the
utility of a player from any bid between $[0,1/2]$ on any item is
equal to $1/2$ and is strictly lower for higher bids.
The expected welfare of this equilibrium is equal to the expected
number of items that are allocated to some bidder. Each item is
allocated to some bidder with probability $3/4$. Thus the total
expected welfare is equal to $3/2$, while the optimal welfare is equal
to $2$.
If the example is extended to $n$ bidders and $n$ items, then the
price of anarchy tends to~$1-1/e$ as $n \rightarrow \infty$.
\end{example}

The primary obstacle to extending the full efficiency 
\etedit{proof (Theorem \ref{thm:sim-fpa-pne})} to mixed Nash
equilibria is the dependence of the proposed deviating bid
on the current bids of the other players.
Any analysis that used deviations that depend on the realization of
others' bids
will only hold for the case of pure Nash equilibria.

One solution would be
to prove an approximate efficiency result that makes use only of
deviations that are independent of others' actions.  But is this even
possible?

The previous section provides an affirmative answer in the case of
first-price single-item auctions.
%
Specifically, in the proof of Theorem~\ref{thm:single-item-fpa},
the deviations used only require that a bidder bid half her value.
This deviation did not depend on others' bids,
and guaranteed a utility at least half of the player's valuation minus the
realization of the item's price, whatever the price may be
(recall~\eqref{eqn:half-value-lower}).

We now extend this idea to multi-item auctions.
Define bidder~$i$'s deviation~$b^*_i$ as bidding half her value on the
item $j^*(i)$ that she receives in some fixed optimal allocation (with
at most one item per bidder), and 0 on all other items.
Following the proof of Theorem~\ref{thm:single-item-fpa},
we have
\begin{equation*}
u_i(b_i^*,\bidsmi;v_i)\geq \frac{v_{ij^*(i)}}{2} - p_{j^*(i)}(\bids)
\end{equation*}
for every bid profile $\bids$ and hence
\begin{equation}\label{eqn:complete-info-smoothness}
\sum_{i\in [n]} u_i(b_i^*,\bidsmi;v_i) \geq \frac{1}{2} \opt(\vals) - \sum_{j\in [m]} p_j(\bids).
\end{equation}
If $\bids$ is a pure Nash equilibrium, then the
inequality~\eqref{eqn:complete-info-smoothness} and the reasoning in
the proof of Theorem~\ref{thm:single-item-fpa}
imply that the social
welfare of $\bids$ is at least half of the maximum possible.

The key point is that, because the deviations $b_i^*$ used in this
derivation are independent of the bid profile $\bids$, the
approximation guarantee applies more generally to mixed Nash equilibria
(and even more generally, see Section
\ref{sec:no-regret}).
To see this, note that each deviation $b_i^*$ is now well defined
even when others' bids are randomized.
Taking the expectation of
\eqref{eqn:complete-info-smoothness} over the mixed Nash equilibrium
bid profile distribution (with the deviations~$b_i^*$ fixed) and using
that players are
best-responding in expectation proves that the expected welfare of
every mixed Nash equilibrium is at least $\tfrac{1}{2}\opt(\vals)$.
The optimized deviations mentioned in Section~\ref{sec:fpa} can be used
to improve the bound from $\tfrac{1}{2}$ to $1-1/e$.
\etedit{\begin{theorem}\label{thm:sim-fpa-ne} Every
complete-information mixed
    Nash equilibrium of the simultaneous first-price auction game with
    unit-demand players achieves expected
    welfare at least $1-1/e$ times the maximum possible.
\end{theorem}}

\subsection{Incomplete information}

The preceding section showed how to cope with randomness in
bidders' strategies and prove efficiency guarantees for
(complete-information) mixed Nash equilibria.
We now consider the incomplete-information case, where the valuation of
each bidder~$i$ is drawn independently from a distribution $\F_i$.
In this setting, a bidder knows her own valuation
and the $\F_i$'s, but not the realizations of others' valuations.
\trdelete{the independent private value Bayesian model that we
introduced and analyzed in Section \ref{sec:fpa} for single-item
first-price auctions, with the valuation of each bidder~$i$ drawn from
a distribution $\F_i$.}
Do our previous price-of-anarchy bounds continue to hold?


Coming up with well-defined deviations $b^*_i$ is again the primary
obstacle to extending our results.
To understand the issue,
recall how we defined $b^*_i$ in the proof of
Theorem~\ref{thm:sim-fpa-ne}, with bidder~$i$ bidding half her
valuation on the item $j^*(i)$ (if any) that she receives in some fixed optimal
allocation.  The optimal allocation, and hence the identity of the
item $j^*(i)$, depend on the full valuation profile $\vals$.  When
valuations are commonly known, bidder~$i$ knows the identity of the
item $j^*(i)$ and is in a position to execute the deviation $b^*_i$.  If
bidder~$i$ only knows a distribution over others' valuations and not
the valuations themselves, then she only knows a distribution over the
possible identities of $j^*(i)$.  The previous deviation~$b^*_i$ is no
longer well defined, derailing the argument.


This obstacle motivates
defining each deviation $b^*_i$ in a way that depends only on the
player's own valuation.
This seems rather restrictive!
Interestingly,
we show \etedit{that} there is an essentially black-box
way to transform each price-independent deviation
used in the proof of Theorem~\ref{thm:sim-fpa-ne}
so that it no longer depends on
the valuations of other bidders,
while at the same time implying the exact same
efficiency guarantee!
The key idea is to use the following (randomized) deviation: a bidder~$i$
samples valuations $\valsmi'$ for the other bidders according to the
(known) valuation distributions and uses $\valsmi'$ as a surrogate
for the true but unknown valuations $\valsmi$.  That is, the bidder
bids half her value on the item $j'(i)$ and 0 on the other items,
where $j'(i)$ is the item $i$ receives (if any) in an optimal
allocation when bidder~$i$ has valuation~$\vali$ and the other bidders
have valuations~$\valsmi'$.  Crucially, this randomized deviation
depends only on bidder~$i$'s valuation (and the \teedit{distributions $\F_{-i}$}), and not on
anyone's bid nor on any other bidder's valuation.
This idea originates in \citeA{Christodoulou2008} and its generality
is made clear in \citeA{Roughgarden2012} and \citeA{Syrgkanis2012}.
It enables us to extend the argument in the proof of
Theorem~\ref{thm:sim-fpa-ne} to establish Theorem~\ref{thm:sim-fpa-bne}.


\begin{proofof}{Theorem \ref{thm:sim-fpa-bne}}
We prove a price-of-anarchy bound of~$\tfrac{1}{2}$.  The improvement
to $1-1/e$
follows similar lines as in the proof of Theorem~\ref{thm:single-item-fpa}.

Denote by $j^*(i,\vals)$ the item awarded to player $i$ in the optimal
allocation for the valuation profile $\vals$.
Define $b_i^*(\vals)$ as in the proof of Theorem~\ref{thm:sim-fpa-ne},
as the bid vector where player $i$ bids half of her
value on item $j^*(i,\vals)$
and zero on every other item.
Inequality~\eqref{eqn:complete-info-smoothness} from the proof
of Theorem~\ref{thm:sim-fpa-ne} implies that,
for every valuation profile $\vals$ and every bid profile $\bids$,
\begin{equation}\label{eqn:sim-fpa-smoothness}
\sum_{i\in [n]} u_i(b_i^*(\vals),\bidsmi;v_i) \geq \frac{1}{2} \opt(\vals) - \rev(\bids),
\end{equation}
where $\rev(\bids)=\sum_{j\in [m]} p_j(\bids)$ denotes the total
revenue of the auction. \teedit{Recall that $b_i^*(\vals)$ is not a valid deviation in the incomplete-information setting, as player $i$ is not aware of the valuations $\valsmi$.}

Consider a Bayes-Nash equilibrium profile of strategies where,
  conditional on her valuation being $v_i$,
player~$i$ chooses a bid according to some distribution $D_i(v_i)$.
For conciseness, denote by $\G_i$ the distribution of player~$i$'s
bid at this equilibrium (drawing $\vali$ from $\F_i$ and then
$b_i$ from $D_i(v_i)$).
Crucially, because
players' valuations are independently distributed, the joint distribution
of equilibrium bids is just the product distribution
$\G_1\times\ldots\times \G_n$, which we denote by~$\G$.  In particular,
the distribution $\G_{-i}$ of the equilibrium bids of players other
than~$i$ is unaffected by conditioning on player~$i$'s valuation
$v_i$.  This would not be the case if
players' valuations were correlated.

Now consider the following valid incomplete-information
  randomized deviation $b_i'\sim D_i'(v_i)$: player $i$ first randomly
  samples a valuation profile $\vals'_{-i} \sim \F_{-i}$ and then performs
  the deviation \teedit{$b_i^*(v_i,\valsmi')$.} 
The Bayes-Nash equilibrium conditions imply that, for every player~$i$
and possible valuation $\val_i$ of the player,
\begin{align*}
\E_{b_i \sim D_i(v_i), \bids_{-i} \sim G_{-i}}\left[u_i(\bids;
  v_i) \right]
& \ge
\E_{b'_i \sim D'_i(v_i), \bids_{-i} \sim
  \G_{-i}}\left[u_i(b'_i,\bidsmi; v_i) \right]\\
& = \E_{\vals'_{-i} \sim \F_{-i}, \bids_{-i} \sim
  \G_{-i}}\left[u_i(b^*_i(v_i,\vals'_{-i}),\bidsmi; v_i) \right].
\end{align*}
This inequality holds for every~$v_i$, and hence also in expectation
over~$v_i$:
\begin{align}\nonumber
\teedit{\E_{v_i \sim \F_i, b_i \sim D_i(v_i), \bids_{-i} \sim G_{-i}}\left[u_i(\bids; v_i) \right]
}
& \ge \E_{\val_i \sim
  \F_i, \vals'_{-i}
  \sim \F_{-i}, \bids_{-i} \sim
  \G_{-i}}\left[u_i(b^*_i(v_i,\vals'_{-i})),\bidsmi; v_i) \right]\\ \label{eq:u}
& =
\E_{\vals \sim \F, \bids \sim \G}
\left[u_i(b^*_i(\vals),\bidsmi; v_i) \right],
\end{align}
where in the equation we have renamed $\vals'_{-i}$ as $\valsmi$ and
also used the fact that the utility of player~$i$ after deviating is
independent of what she would bid at equilibrium.  We emphasize that,
on the right-hand side of~\eqref{eq:u}, $\vals$ and $\bids$ are drawn independently
from $\F$ and $\G$, respectively.

On the other hand, taking the expectation of inequality
\eqref{eqn:sim-fpa-smoothness} over $\vals \sim \F$ and $\bids \sim
\G$ shows that
the sum of the expected deviating utilities across players is
\begin{equation}\label{eqn:sim-fpa-fixed-ac-sum}
\sum_{i\in [n]} \E_{\vals \sim \F, \bids \sim \G}\left[
  u_i(b_i^*(\vals),\bidsmi;v_i) \right] \geq \frac{1}{2} \E_{\vals \in
  \F}\left[\opt(\vals)\right] -
\E_{\bids \sim \G}\left[\rev(\bids)\right].
\end{equation}
Summing inequality~\eqref{eq:u} over the bidders~$i$ \teedit{and} combining \teedit{it} with
inequality~\eqref{eqn:sim-fpa-fixed-ac-sum}, \teedit{we get
$$
\E_{\vals \sim \F, b_i\sim D_i(v_i) \forall i}\left[\sum_i u_i(\bids; v_i) \right]\ge \frac{1}{2} \E_{\vals \in
  \F}\left[\opt(\vals)\right] -
\E_{\bids \sim \G}\left[\rev(\bids)\right].
$$
Now} adding $\E_{\bids  \sim \G}\left[\rev(\bids)\right]$ to both sides shows that the expected welfare of the Bayes-Nash equilibrium is at least $\tfrac{1}{2}$ times the expected maximum welfare.
\end{proofof}


\section{General Auctions and Smoothness}
\label{sec:model}

The examples in the previous sections portrayed how we can bypass the
daunting task of characterizing the equilibria of a game of incomplete
information and directly show that every equilibrium is
approximately efficient. Our next goal is to develop
a general framework for
providing such efficiency guarantees, building on our previous arguments
for single- and multi-item first-price auctions.

\subsection{General Auction Mechanisms}

We begin by formally defining a general mechanism design setting and a
generic auction, introducing some essential notation used
throughout the survey.  In a general mechanism design setting, the
auctioneer solicits an action $\ac_i$ from each player $i$ from
some action space~$\A_i$ \teedit{(a bid in the case of the auctions considered in the previous two sections)}. Define
$\A=\A_1\times\ldots\A_n$. Given the action profile
$\acs=(\ac_1,\ldots,\ac_n)\in \A$, the auctioneer decides an outcome
$o(\acs)$ among a set of feasible outcomes $\Ou$. Part of an outcome is
also a payment $p_i(o)$ that the auctioneer receives from each
player. We denote by $\rev(o) =\sum_i p_i(o)$ the revenue of
the auctioneer.

Each player derives some utility which is a function of the outcome
and of a parameter~$\val_i$ taking values in a parameter space $\V_i$,
typically referred to as the {\em valuation} or {\em type} of the
player. We
denote by $u_i(o;\val_i)$ the utility of a player with type
$\val_i\in \V_i$ in the outcome $o\in \Ou$.  Write $\V$ for the set
$\V_1 \times \V_2 \times \cdots \times \V_n$ of valuation profiles.

For a given auction, since the outcome is uniquely defined by the
action profile, we overload notation and
write $u_i(\acs; \val_i)$ for the utility of the player
with type $\val_i$
\etedit{for
  the outcome} in the auction under an action profile $\acs$.
We also denote by $\rev(\acs)$ the revenue of the
auctioneer in action profile $\acs$.

\paragraph{A few additional examples.} \etedit{Before formally
  defining equilibria and smoothness and proving price-of-anarchy
  bounds, we give a few additional examples of auctions of
  interest that fit in the framework just defined.} \teedit{While the general definition allows the player utility $u_i(\acs; \val_i)$ to depend on the price paid in more complex ways, in all of our examples we use quasi-linear utility, where $u_i(\acs; \val_i)=v_i(\acs)-p_i(\acs)$, for some value function  $v_i(\acs)$ that depends only on the allocation of the auction, and $p_i(\acs)$ is the price paid by player $i$.}

\etedit{Our first additional} example is
again for the sale
of a single item, \etedit{but with the
  different payment rule in which \emph{all} players pay their bid,
  and not only   the winner}. 
\etedit{Such} {\em all-pay auctions} arise naturally in contest
settings where players have sunk costs that they have to pay
irrespective of whether or not they win a prize \etedit{(e.g., where
  the work invested acts as the payment)}. There are many works
in the economic literature
\cite{Amann1996,Baye1996,Gneezy2006,Siegel2014} and a few recent ones
in computer science \cite{Dipalantino2009,Chawla2012} that analyze
properties of \etedit{all-pay auctions.} 
As with the first-price auction, the asymmetric independent private
values model has been under-studied because of the difficulty of
characterizing its equilibria in a closed form.

\begin{example}{All-pay auction}\label{ex:all-pay}
  Consider a setting where $n$ players bid for a single item. Each
  player has a value $v_i$ for the item, drawn from a commonly known
  distribution $\F_i$. Each player submits a bid $b_i$. The highest
  bidder wins the
  item and every player pays her bid (whether a winner or not).
\end{example}

\etedit{A different example arises in the context of public goods,
such as a bridge or a park. A group of people
  needs to decide whether to build such a public good, and how to
  share the cost. This scenario is especially tricky when the public
  good is non-excludable, meaning that everyone can use it, whether or
  not they contributed to its construction (as with most bridges and
  parks).}

\etedit{
\begin{example}{Public Good}\label{ex:public-good}
A group of $n$ players bid on their share of a joint
project that has a publicly known cost $c$. Each player has a value
$v_i$ for the completion of the project.
Each player submits a bid $b_i$. If $\sum_i b_i \ge c$,
the project is undertaken, and all players pay their bids, so
bidder~$i$'s utility for the outcome is $v_i-b_i$. If $\sum_i b_i < c$,
then the project is abandoned and payments are \teedit{not} collected.
\end{example}
}

\paragraph{Incomplete information and equilibria.} The type $\val_i$ of each player is private information and is drawn independently from a commonly known distribution $\F_i$. \etedit{We denote by $\A_i$ the possible actions of the player, and by $\Delta(\A_i)$ the probability distributions over these actions.}
The strategy of a player in this auction is a mapping $\sr_i: \V_i\rightarrow \Delta(\A_i)$ from a type $\val_i\in \V_i$ to a distribution over actions $\sr_i(\val_i) \in \Delta(\A_i)$.
\begin{defn}[Bayes-Nash equilibrium ($\BNE$)] A strategy profile
  $\sigma=(\sigma_1,\ldots,\sigma_n)$ is a
{\em Bayes-Nash equilibrium} if: for each player $i$, for each type $\val_i\in \V_i$, and for each action $\ac_i'\in \A_i$:
\begin{equation*}
\E_{\valsmi\sim \F_{-i}}\left[ \E_{\acs\sim \sigma(\vals)}\left[u_i(\acs; \val_i)\right] \right] \geq \E_{\valsmi\sim \F_{-i}}\left[ \E_{\acs_{-i}\sim \sigma_{-i}(\valsmi)}\left[u_i(\ac_i',\acs_{-i}; \val_i)\right]\right].
\end{equation*}
\end{defn}


\paragraph{Social welfare.} We will be interested in analyzing the
social welfare, 
defined as the utility of all
participating parties (the bidders and the auctioneer):
\begin{equation}\label{eq:sw2}
SW(o;\vals) = \sum_{i=1}^{n} u_i(o;\val_i) + \rev(o).
\end{equation}
For players with quasi-linear utilities, as in~\eqref{eq:qlsi}
and~\eqref{eq:ql-sim}, this definition coincides with that
in~\eqref{eq:sw}.
The definition in~\eqref{eq:sw2}
makes sense with arbitrary player utility functions,
not just quasi-linear utility functions.
For a given valuation profile $\vals$, we denote the optimal welfare
of a feasible outcome by
$\opt(\vals) = \max_{o \in \Ou} SW(o;\vals)$.

We measure the inefficiency of Bayes-Nash equilibria of an auction
with the {\em price of anarchy (PoA)},
defined as
\begin{equation*}
\poa = \inf_{\F,~\sr \text{ is } \BNE}
\frac
{\E_{\vals\sim \F}\left[\E_{\acs \sim \sr(\vals)}\left[SW(\acs;\vals)\right]\right]}{\E_{\vals\sim \F}\left[\opt(\vals)\right]}.
\end{equation*}
The \poa is between 0 and 1, with numbers closer to~1 corresponding to
more efficient equilibria.


\subsection{Smooth Auctions}
\label{sec:smoothness}

We begin our analysis of general auction settings by observing
that all of the proofs of the efficiency guarantees in Sections
\ref{sec:fpa} and \ref{sec:sim-fpa} follow the exact same paradigm.
The key step is to
find an appropriate deviation \etedit{($b_i^*$, or more
  generally an action} $a_i^*$) for each player, such that no matter
what others' actions are, the utility achieved by the deviation can be
bounded below by some fraction of the player's contribution to the
optimal welfare, less some quantity that relates to the current revenue
of the auction. (In fact,
we only needed this in aggregate over all of the players.) This was
crystalized by Equation \eqref{eqn:example-smoothness} for the case of
Bayes-Nash equilibria of single-item first-price auctions, Equation
\eqref{eqn:pne-smoothness} for complete-information pure Nash
equilibria of simultaneous first-price auctions, and Equation
\etedit{\eqref{eqn:complete-info-smoothness}} 
for mixed Nash equilibria of such auctions. All of these inequalities
are variants of what we will call a \emph{smoothness-type}
inequality.

The only difference in the different versions of the argument is
the informational requirements of the deviations $a^*_1,\ldots,a^*_n$.
For instance, in Equation \eqref{eqn:pne-smoothness} each
\etedit{$b_i^*$} 
depends on the valuations and bids of all of the players; this
restricts the analysis to complete-information pure Nash equilibria.
In Equation
\etedit{\eqref{eqn:complete-info-smoothness}}, 
each $b_i^*$ depends only on the valuations of players and not on
their bids; this is sufficient to analyze mixed Nash equilibria and,
after the application of a ``black-box translation'' to remove the
dependence of a bidder's deviation on others' valuations, Bayes-Nash
equilibria in the incomplete-information setting with independent
valuation distributions.  In Equation \eqref{eqn:example-smoothness},
each \etedit{$b_i^*$} 
depends on the player's own valuation and on nothing else, and the
consequent approximate efficiency bound applies even to Bayes-Nash
equilibria of first-price single-item auctions with correlated bidder
valuations.

Our primary goal in this section is to understand the efficiency of
Bayes-Nash equilibria of auctions in the independent private
values model.  This motivates defining a
{\em smooth} auction as one that satisfies a natural
generalization of Equation \etedit{\eqref{eqn:complete-info-smoothness}}, 
with the restriction that the deviations used do not depend on
players' actions (but can depend on all players' valuations).

\trdelete{
In Section~\ref{sec:sim-fpa} we explain why Equation
\eqref{eqn:pne-smoothness} is not a strong enough property to
conclude approximate efficiency guarantees for the incomplete-information
case.  The analysis based on \eqref{eqn:pne-smoothness}
does not even carry over to mixed Nash equilibria in the
complete-information setting.
Thus allowing each deviating action to depend on all
players' actions seems too permissive and any conclusion on
approximate efficiency drawn from such an analysis will apply only to
the pure Nash equilibria of the auction. On the other hand, even though
Equation \eqref{eqn:example-smoothness} is strong enough to prove
efficiency guarantees for the incomplete-information model, even with
correlated
valuations (see Theorem \ref{thm:single-item-corr-fpa}), the
restriction that the deviation can only depend on the player's
own valuation is restrictive and ties the analyst's hands rather
tightly. For instance, the deviations that we used for mixed Nash
equilibria of simultaneous first-price auctions heavily depended on
other players' valuations.
}
\trdelete{
However, as we showed in the previous section, even if we allow the
deviations to depend on other players' valuations, there exists an
essentially black-box way of translating such deviations into
deviations that do not have such a dependence
(see the proof of Theorem \ref{thm:sim-fpa-bne} and the preceding
discussion).
Thus this
type of restriction on the deviations used, seems to combine almost
the best of both worlds: (1) it is permissive enough for the
\emph{lazy   analyst}
to only have to argue in the complete information setting,
(2) it is restrictive enough for the conclusions based on such an
analysis to carry over to the incomplete information setting. Due to
this fact,
}

\begin{defn}[Smooth auction]\label{def:smooth-auction}
For parameters $\lambda \ge 0$ and $\mu \ge 1$,
an auction is
  {\em $(\lambda,\mu)$-smooth} if for every valuation profile $\vals
  \in \V$ there exist action distributions
  $D^*_1(\vals),\ldots,D^*_n(\vals)$ over $\A_1,\ldots,\A_n$
such that, for every action
  profile $\acs$,
\begin{equation}\label{eqn:general-smoothness}
\sum_i \E_{a_i^*\sim D_i^*(\vals)}\left[u_i(a_i^*,\acs_{-i}; v_i)\right] \geq \lambda \opt(\vals) - \mu \rev(\acs).
\end{equation}
\end{defn}
Definition~\ref{def:smooth-auction} is due to \citeA{Syrgkanis2013},
inspired by previous definitions of ``smooth games'' in
complete-information \cite{Roughgarden2009} and
incomplete-information \cite{Roughgarden2012,Syrgkanis2012} settings.
The realizable smoothness parameters of a fixed auction format
generally depend on the class~$\V$ of permitted valuation profiles;
this point is particularly important in Section~\ref{sec:composability}.

On the surface, the inequality~\eqref{eqn:general-smoothness} appears
relevant only for case of complete information (since $\vals$ is
fixed in~\eqref{eqn:general-smoothness}) and pure Nash equilibria
(since $\acs$ is similarly fixed).  But the ideas we developed for
simultaneous first-price auctions (Section~\ref{sec:sim-fpa}) indicate
how to extend efficiency guarantees when actions or valuations are
randomized.
For example,
the proof of Theorem~\ref{thm:sim-fpa-ne} generalizes without
difficulty to all smooth auctions.

\begin{theorem}\label{thm:smooth-Nash}
If an auction is $(\lambda,\mu)$-smooth,
then for every valuation profile~$\vals \in \V$,
every complete-information mixed Nash equilibrium of the
auction has expected welfare at least $\frac{\lambda}{\mu} \cdot
\opt(\vals)$.
\end{theorem}

\trdelete{
\vsedit{Along similar lines as in the proof of Theorem \ref{thm:sim-fpa-bne} one can show that the almost black-box extension from complete information to incomplete information always holds in the general auction setting, when the auction is $(\lambda,\mu)$-smooth. The proof is based on the idea of considering the incomplete information deviation which first randomly samples a profile $v_{-i}$ of opponents values from $\F_{-i}$ and then draws an action from the complete information smoothness deviation $D_i^*(v)$.}
}

Similarly, the proof of Theorem~\ref{thm:sim-fpa-bne} generalizes to
all smooth auctions.

\begin{theorem}[Extension to Incomplete Information]\label{thm:extension-incomplete}
If an auction is $(\lambda,\mu)$-smooth,
then for every profile $\F_1,\ldots,\F_n$ of independent valuation
distributions over $\V_1,\ldots,\V_n$,
every Bayes-Nash equilibrium of the
auction has expected welfare at least $\frac{\lambda}{\mu} \cdot
\E_{\vals \sim \F}[\opt(\vals)]$.
\end{theorem}

Results like Theorems~\ref{thm:smooth-Nash}
and~\ref{thm:extension-incomplete} are sometimes called ``extension
theorems,'' meaning that they extend an approximate efficiency
guarantee from a restricted class of equilibria (like
complete-information pure Nash equilibria) to a more general class
(like Bayes-Nash equilibria).  They free the analyst to focus on
proving inequalities of the form~\eqref{eqn:general-smoothness},
without concern for randomness in valuations or in actions.
Proofs that establish inequalities of the
form~\eqref{eqn:general-smoothness}, for a suitable $\lambda$ and
$\mu$, are sometimes called ``smoothness proofs.''


\subsection{Another Example: All-Pay Auctions}

To portray the generality of the smoothness approach, \etedit{we
  discuss how it applies to all-pay auctions (Example \ref{ex:all-pay}).}
\etdelete{we give another application for the case when a single item
  is being auctioned via an all-pay auction. All-pay auctions arise
  naturally in contest settings where players have sunk costs that
  they have to pay irrespective of whether they end up winning a prize
  or not. There is large literature in economics
  \cite{Amann1996,Baye1996,Gneezy2006,Siegel2014}, and recently in
  computer science \cite{Dipalantino2009,Chawla2012} analyzing
  properties of such types of auctions. Similar to the first price
  auction, the asymmetric independent private values model has been
  largely under-studied due to the difficulty of characterizing its
  equilibria in a closed form.
%
Consider a setting where $n$ players bid for a single item. Each
player has a value $v_i$ for the item drawn from a distribution
$\F_i$. Each player submits a bid $b_i$. The highest bidder wins the
item and all players pay their bid irrespective of whether they won
the item or not.}
\etedit{The extension theorems are especially interesting in this case,
  as all-pay auctions do not have pure Nash equilibria.  Despite this,
our results allow the analyst to think only about a property
  of pure strategies, and then conclude an approximate efficiency
  guarantee for Bayes-Nash equilibria.}

We show that the all-pay auction is $(1/2,1)$-smooth.
Thus for every valuation profile $\vals$,  we require
a deviation $b^*_i$ for each player (possibly randomized) such that
for every bid profile $\bids$,
\begin{equation}
\sum_{i\in [n]} \E_{b_i^*}\left[u_i(b_i^*,\bidsmi;v_i)\right] \geq \frac{1}{2} \opt(\vals) - \rev(\bids).
\end{equation}
Without loss of generality, assume that player $1$ has the highest
valuation in the profile $\vals$.  Our deviating bids are:
$b_1^*$ is chosen uniformly at random from $[0,v_1]$, and $b^*_i = 0$
for all $i > 1$.
Now fix a bid profile $\bids$.
Player $1$ wins the item after deviating to $b_1^*$ whenever this bid is
above the highest bid in $\bids_{-1}$, in which case she gets a value of
$v_1$.
Her expected payment is exactly $v_1/2$ (recall it's an all-pay
auction). Thus her expected utility after deviating can
be bounded below as follows:
\begin{align*}
\E_{b_1^*}\left[u_i(b_1^*,\bids_{-1};v_1)\right] =~& v_1\cdot
                                                     \Pr\left[b_1^*>\max_{j
                                                     > 1}b_j\right] - \frac{v_1}{2}
\geq v_1 \cdot \frac{v_1-\max_{j}b_j}{v_1} - \frac{v_1}{2} = \frac{v_1}{2} - \max_{j} b_j.
\end{align*}
The utility of every other player from the deviation $b_i^*=0$ is
non-negative.  Summing all these lower bounds and observing that
$\opt(\vals)= v_1$ and $\rev(\bids) \geq \max_{j} b_j$ verifies the
inequality~\eqref{eqn:general-smoothness} with $\lambda =
\tfrac{1}{2}$ and $\mu = 1$.

Combining this smoothness proof with
Theorem \ref{thm:extension-incomplete}, we conclude that every Bayes-Nash equilibrium of an asymmetric all-pay auction \teedit{with players with independent types} achieves expected welfare at least half of the expected optimal welfare.

\subsection{Correlated Valuations}

The guarantee in Theorem~\ref{thm:extension-correlated} is for
independently drawn player valuations --- what if valuations are
correlated?
Unfortunately, the guarantee no longer holds:
there exist auctions that satisfy
Definition~\ref{def:smooth-auction} for constant $\lambda$ and $\mu$,
but which have unbounded inefficiency with correlated valuations as
the number of players and
items in the market grows
\cite{Bhawalkar2011,Feldman2013}.

On the positive side,
Theorem \ref{thm:single-item-corr-fpa} shows
that the Bayes-Nash equilibria of single-item first-price auctions are
approximately efficient even when players' valuations are
correlated.
The crucial property that enables this result is that the single-item
first-price auction satisfies a smoothness-type inequality where the
deviating action depends only on the player's own valuation (and not
on others' valuations).
This type of
smoothness property was defined by \citeA{Lucier2011}, which they called
\emph{semi-smoothness}. Since the only difference with Definition
\ref{def:smooth-auction} is the independence of the deviating action
from others' valuations, we
refer to this property as
\emph{smoothness with private deviations}.

\begin{defn}[Smooth Auction with Private
  Deviations]\label{def:semi-smooth-auction}
For parameters $\lambda \ge 0$ and $\mu \ge 1$,
an auction is
  {\em $(\lambda,\mu)$-smooth with private deviations} if for every
  valuation profile $\vals
  \in \V$ there exist action distributions
  $D^*_1(\val_1),\ldots,D^*_n(\val_n)$ over $\A_1,\ldots,\A_n$
such that, for every action
  profile $\acs$,
\begin{equation}
\sum_i \E_{a_i^*\sim D_i^*(\val_i)}\left[u_i(a_i^*,\acs_{-i}; v_i)\right] \geq \lambda \opt(\vals) - \mu \rev(\acs).
\end{equation}
\end{defn}

The requirement in Definition~\ref{def:semi-smooth-auction} is
considerably stronger than that in Definition
\ref{def:smooth-auction}, but the reward is approximate efficiency
guarantees with correlated valuations.
\begin{theorem}[Extension to Correlated
  Valuations]\label{thm:extension-correlated}
If an auction is $(\lambda,\mu)$-smooth,
then for every joint distribution $\F$ over players' valuations,
every Bayes-Nash equilibrium of the
auction has expected welfare at least $\frac{\lambda}{\mu} \cdot
\E_{\vals \sim \F}[\opt(\vals)]$.
\end{theorem}

Because of the stronger condition in
Definition~\ref{def:semi-smooth-auction}, the proof of
Theorem~\ref{thm:extension-correlated} is simpler than that of
Theorem~\ref{thm:extension-incomplete} --- one can just use
each deviation $D^*_i(\val_i)$ directly, and there is no need to
randomly sample fictitious valuations for the other players (and hence
no need for independent valuations).
The proof follows the same lines as that of
Theorem~\ref{thm:single-item-fpa}; see
\citeA{Lucier2011} or \citeA{Roughgarden2012} for the details.

\trdelete{
 \begin{example}{Generalized First Price Auction}
\vscomment{It might be interesting to add here the Ad auctions example which seems of practical importance and a nice portrayal of interesting settings where the stronger smoothness is satisfied.}
 \end{example}
}

\section{No-Regret Learning}
\label{sec:no-regret}

So far in this survey we have only analyzed auctions at equilibrium.
But how do players arrive at this
equilibrium?  In many real-world applications of auction design,
players do not participate in the auction only once and then
vanish. Typically they participate in the auction repeatedly.

The Bayes-Nash equilibrium condition implicitly assumes that
when a player arrives in a market, they have done their homework well: they have formed their
beliefs about the competition and have computed a Bayes-Nash
equilibrium of the market for these beliefs.
On arrival to the market,
they simply invoke the equilibrium strategy that
they have pre-computed for their realized value.

In many auction settings,
this
assumption that players are such diligent students is a strong one,
especially in cases where \teedit{entering players lack information about the environment, or when} the problem of computing an equilibrium is
computationally hard.  
\teedit{The expectation of such diligent preparation is even more unreasonable} in repeated
auction environments
where the stakes of each individual
auction are small, with the aggregate payoff over time being of
primary importance. A more reasonable assumption is that
players experiment in the market and try to optimize their bid
over time using their past experience as a proxy for future
rewards. For example, in Internet advertising auctions, there are a
number of different (adaptive) bidding agents available.
In general, the study of adaptive game playing is
called
\emph{the theory of learning in games}
\cite{Fudenberg1998}.

Can we bound the average efficiency of an auction when players use
adaptive game-playing algorithms? Do the equilibrium efficiency
guarantees that we provided thus far extend to such adaptive game
playing?

One attractive model of adaptive game playing is \emph{no-regret
  learning} \cite{Freund1999,Hart2000,Auer1995}, which dates back to
the very early work of \citeA{Hannan1957}. No-regret learning has a
long and distinguished history even outside of game theory, when a
single decision maker is facing a sequence of decisions among a fixed
set of actions whose rewards at each time step are chosen by an
adversary (see e.g.\ \citeA{Cesa-Bianchi2006} for an extended
survey). It is easy to see the relevance of this model to a repeated
game
environment: instead of an adversary, the reward of each action at
each iteration is affected by the actions of other players.
When the other players are hard to predict, the player might as well
treat them as adversarial.

A learning algorithm for a player satisfies the no-regret condition if,
in the limit as the number of times the game is played goes to
infinity, the average reward of the algorithm is at least as good as
the average reward of the best fixed action in hindsight (assuming the
\teedit{sequence of} actions of the other players remain
\etedit{unchanged}). 
Many simple algorithms are known to achieve this property, including the
\emph{multiplicative weight updates} algorithm
\cite{Littlestone1994,Freund1999} and the \emph{regret matching}
algorithm \cite{Hart2000}.  See \citeA{Cesa-Bianchi2006} for more
general classes of no-regret algorithms.

This section addresses whether the price-of-anarchy
guarantees of the previous sections extend to the average welfare of
repeated auctions, assuming that each player uses a no-regret learning
algorithm. We will argue that the price of anarchy bound of
$\frac{\lambda}{\mu}$ for a $(\lambda,\mu)$-smooth auction
(Theorems~\ref{thm:smooth-Nash}--\ref{thm:extension-incomplete})
directly
extends to such no-regret learning outcomes, thereby providing further
robustness \etedit{for the welfare} properties of smooth auctions.

More formally, we consider an auction with $n$ players that is
repeated for $T$ time steps. Each player $i$ has some fixed
valuation\footnote{This valuation can be thought of as being drawn at the
  beginning of time from the distribution $\F_i$. Recent work of
  \citeA{Hartline2015} extends the results of this section to the case where
  player~$i$'s valuation is drawn at each iteration from $\F_i$, rather than
  being fixed.} $v_i$ and at each iteration $t$, she chooses to submit
some action $a_i^t$ which can depend on the history of play. After
each iteration, each player observes the actions taken by the other
players.\footnote{This assumption can be relaxed,
and the theory also \teedit{extends} 
to a ``bandit'' model where
each player only observes the utility of the action taken.}

If a player~$i$ uses a no-regret learning algorithm, then
in hindsight her average regret for any alternative strategy $a_i'$
goes to zero \teedit{or becomes negative} (as $T \rightarrow \infty$).
When every player uses such an algorithm, the result is
a {\em vanishing regret} sequence.

\begin{defn}[Vanishing Regret]\label{d:vanishing}
A sequence of action profiles $\acs^1,\acs^2,\ldots,$ is a {\em vanishing regret sequence} if for every player $i$ and action $a_i'\in \A_i$,
\begin{equation}
\lim_{T\rightarrow \infty} \frac{1}{T} \sum_{t=1}^T \left(u_i(a_i',\acs_{-i}^t;v_i)-u_i(\acs^t;v_i)\right) \le 0.
\end{equation}
\end{defn}

We now argue that if an auction is $(\lambda,\mu)$-smooth, then
eventually the average welfare of every vanishing regret sequence is
at least $\frac{\lambda}{\mu}$ times the optimal welfare. The proof is
not hard and we sketch it here.
For simplicity, suppose that each player after $T$ time steps
already has zero regret (or less) for each action.  In particular,
each player~$i$ has no regret with respect to the randomized
action~$a_i^* \sim D^*_i(\vals)$ prescribed by the smoothness property (Definition~\ref{def:smooth-auction}).
This implies that
\begin{equation*}
\frac{1}{T} \sum_{t=1}^T u_i(\acs^t;v_i) \geq \frac{1}{T} \sum_{t=1}^T
\E_{a_i \sim D^*_i(\vali)} \left[u_i(a_i^*,\acs_{-i}^t;v_i)\right].
\end{equation*}
Summing this inequality over all players and invoking the
smoothness inequality~\eqref{eqn:general-smoothness}
for each action profile $\acs^t$, we can easily conclude that
$\frac{1}{T}\sum_{t=1}^T SW(\acs^t;\vals) \geq \frac{\lambda}{\mu}\opt(\vals)$.
The same reasoning straightforwardly yields
the following result for vanishing regret sequences.

\begin{theorem}[Extension to Vanishing Regret
  Sequences]\label{thm:extension-regret}
If an auction is $(\lambda,\mu)$-smooth,
then for every valuation profile~$\vals$,
every vanishing regret sequence of the
auction has expected welfare at least $\frac{\lambda}{\mu} \cdot
\opt(\vals)$ as $T\rightarrow \infty$.
\end{theorem}


\section{Composability}
\label{sec:composability}

\begin{quotation}
{\em ``Most analyses of competitive bidding situations are based on the assumption
that each auction can be treated in  isolation. This assumption is sometimes
unreasonable.''} \cite{Milgrom1982}
\end{quotation}

This section gives a general approach for analyzing the efficiency of
multiple auctions that take place simultaneously, when players have
valuations that are complex functions of the outcomes of the different
auctions.
Specifically, we prove a ``composition theorem''
stating
that, under a ``complement-free'' assumption on players' utility
functions, the simultaneous composition of smooth auctions is again
smooth.


\paragraph{Simultaneous composition.}
We consider a setting with $n$ bidders and $m$ auctions. Each
auction $j$ concerns its own set of items, its own feasible outcome
space $\Ou_j$, for
each player~$i$ its own action space $\A_{ij}$, and its own
outcome function $o_{j}(\acs_j)$.
Similarly, we denote by
$p_{ij}(o_j)$ the payment of player $i$ in auction $j$ in some outcome
$o_j$ and by $\rev_j(o_j)=\sum_{i}p_{ij}(o_j)$ the revenue of
auctioneer $j$.
For example, if each auction is a first-price single-item auction,
then each $\A_{ij}$ is just the set of possible bids (i.e.,
$\mathbb{R}^+$),  $o_j(\acs_j)$ awards
item~$j$ to the player who bid the highest for it, and $\rev_j(o_j)$
equals the highest bid on~$j$.
The {\em simultaneous composition} of the $m$ auctions is the auction
in which each player~$i$ simultaneously picks an action~$a_{ij}$ for
every auction~$j$, resulting in the outcome $\os=(o_1(\acs_1),\ldots,o_m(\acs_m))$.

\paragraph{Utilities and valuations.}
Importantly, a player's utility is now a function
$u_i(\os;v_i)$ of the
outcomes $\os=(o_1,\ldots,o_m)$ of all of the auctions,
where $v_i\in \V_i$ is the player's type.
In the most commonly studied utility model,
each player has a quasi-linear utility function
and is indifferent about what happens to the other players.
Mathematically, this translates to
\begin{equation}\label{eq:ql}
u_i(\os; v_i) = v_i(S_i(\os)) - \sum_j p_{ij}(o_j),
\end{equation}
where $S_i(\os)$ denotes the items awarded to player~$i$ in the
outcome $\os$, and $v_i$ is a {\em valuation function} that
assigns a nonnegative value $v_i(S)$ to each subset~$S$ of items that
player~$i$ might obtain.
For example, if every auction is a single-item auction, then
$S_i(\os)$ corresponds to the auctions for which~$i$ is the
winner.

No positive results are possible with arbitrary valuation functions
(see Section~\ref{sec:lower_bounds}), so to make progress we need to
impose structure on players' valuations.
The simplest type of valuation function~$v_i$ is an {\em additive}
function, where there are nonnegative ``item valuations''  $v_{i1},\ldots,v_{im}$
such that, for every bundle $S$ of items, the valuation $v_i(S)$ for
the bundle is just the sum $\sum_{j \in S} v_{ij}$ of the
corresponding item valuations.
A player with an additive valuation can reason about each item
separately.
The {\em unit-demand} valuations introduced in
Section~\ref{sec:sim-fpa} are another example.  Here, for nonnegative
item valuations $v_{i1},\ldots,v_{im}$, the corresponding unit-demand
valuation is $v_i(S) = \max_{j \in S} v_{ij}$.
Already in this case, the value of a player for some item $j$ might
depend on whether or not she wins a different item $j'$ sold in a
different auction.

Both additive and unit-demand valuations are special types of {\em
  monotone submodular} valuations, meaning that $v_i(S) \le v_i(T)$
whenever $S \sse T$ and
$
v_i(T \cup \{j\}) - v_i(T) \le v_i(S \cup \{j\}) - v_i(S)
$
for every $S \sse T$ and item~$j$ \teedit{also, referred to as having the decreasing marginal value property}.  Such functions model diminishing
returns, and have been studied extensively in machine learning,
optimization, and economics.

It is technically convenient to work with a still more general class
of valuations, which we call {\em complement-free} valuations.
Intuitively, the following definition prohibits complementarities
between different items, where winning one item is valuable only if
the bidder also wins some other item.\footnote{In the literature,
  ``complement-free'' is often equated with subadditivity.  What we
  are calling ``complement-free'' is often called ``fractionally
  subadditive'' or ``XOS'' (for ``XOR-of-singletons'') in the
  literature.  Valuations that are complement-free in the sense of
  this survey are always subadditive, but not conversely.
See~\citeA{Lehmann2001} for further details.}
Precisely, a valuation function $\vali$ is {\em complement-free}
if there exist additive valuations $\a^1_i,\ldots,\a^r_i$ such
that, for every subset $S$ of items,
\begin{equation}\label{eq:xos}
\vali(S) = \max_{\ell=1}^r \left\{ \sum_{j \in S} w^{\ell}_{ij}
\right\}.
\end{equation}
For example, if $\vali$ is a unit-demand valuation with item
values $v_{i1},\ldots,v_{im}$, then $\vali$ is the maximum of $m$
additive valuations $\a^1_i,\ldots,\a^m_i$, where $w^{\ell}_{ij}$
equals $v_{ij}$ if $\ell = j$ and 0 otherwise.
More generally, every monotone submodular valuation is complement-free.
The proof uses $m!$ additive valuations, one for each ordering $\pi$
of the $m$ items, and defines $w^{\pi}_{ij}$ as $v_i(S^{\pi}_j \cup
\{j\}) - v_i(S^{\pi}_j)$, where $S^{\pi}_j$ denotes the set of items
that precede $j$ in $\pi$.  The condition~\eqref{eq:xos} holds:
for every subset $S$ of items, $\a_i^{\pi}(S) \le v_i(S)$ for
every $\pi$ (by submodularity), and there exists a $\pi$ such that
$\a_i^{\pi}(S) = v_i(S)$ (take a $\pi$ where the items of $S$ come
first).


We can analogously define a utility function $u_i$ over outcomes $\os
\in \Ou$ as {\em additive} if it has the form
\[
u_i(\os; \val_i) = \sum_{j=1}^m u_{ij}(o_j; \val_{ij})
\]
for valuations of the form $v_i = (v_{i1},\ldots,v_{im})$, and {\em
  complement-free} if it is the maximum of additive utility functions:
\begin{equation}\label{eq:cfutils}
u_i(\os; \val_i) = \max_{\ell=1}^r
\left\{ \sum_{j=1}^m u_{ij}(o_j; \val^{\ell}_{ij}) \right\}.
\end{equation}
For example, with single-item auctions and quasi-linear
utilities~\eqref{eq:ql}, additive and complement-free valuations
induce additive and complement-free utility functions, respectively
(with $u_{ij}(o_j; \val^{\ell}_{ij})$ equal to
$\val^{\ell}_{ij}x_{ij}(o_j)-p_{ij}(o_j)$, where $x_{ij}(o_j)$
indicates whether or not~$i$ wins the $j$th item in $o_j$).
The composition theorem below only requires that players' utilities
are complement-free in the sense of~\eqref{eq:cfutils}, and does not
assume quasi-linear utilities or single-item auctions.


\paragraph{The Composition Theorem.}
The main result of this section is the following composition theorem.
\begin{theorem}[Composition Theorem]\label{thm:composition}
If players have complement-free utility functions, then the
simultaneous composition of $(\lambda,\mu)$-smooth auctions is again a
$(\lambda,\mu)$-smooth auction.
\end{theorem}

\begin{proof}
As a warm-up, suppose that every player~$i$ actually has an additive
utility function, with valuation~$v_i = (v_{i1},\ldots,v_{im})$.
Intuitively, in this case the composition
theorem follows just by adding up the smoothness inequalities for the
constituent auctions.

Formally,
write $\vals_j$ for the projection $(v_{1j},\ldots,v_{nj})$ of $\vals$
onto the valuation space for the $j$th auction.
To save notation, we write $u_i(\acs; v_i)$ to mean $u_i(o(\acs); v_i)$,
where $o(\acs)$ is the outcome resulting from the action profile $\acs$.
Since each auction~$j$ is $(\lambda,\mu)$-smooth (Definition~\ref{def:smooth-auction}), there are action
distributions $D^*_{1j}(\vals_j),\ldots,D^*_{nj}(\vals_j)$
over $\A_{1j},\ldots,\A_{nj}$ such that, for every action profile
$\acs_j \in \A_{1j} \times \cdots \times \A_{nj}$,
\begin{equation}\label{eq:single}
\sum_i \E_{a_{ij}^*\sim D_{ij}^*(\vals_j)}\left[u_{ij}(a_{ij}^*,\acs_{-i,j};
  v_{ij})\right] \geq
\lambda \opt_j(\vals_j) - \mu \rev_j(\acs_j),
\end{equation}
where $\opt_j(\vals_j)$ denotes the maximum social welfare (i.e., sum of
players' utilities plus seller revenue) of any outcome of auction~$j$
with valuation profile $\vals_j$.
Now define $D^*_i(\vals) = D^*_{i1}(\vals_1) \times \cdots \times
D^*_{im}(\vals_m)$.  Using the additivity of players' utilities and
summing up~\eqref{eq:single} over all~$j$ verifies the smoothness
condition~\eqref{eqn:general-smoothness} for the composition of the
auctions.

For the general case, fix a profile $\vals$ of valuations such that
all players' have complement-free utility functions.
The idea is to reduce to the additive case.
To extract additive utility functions from the given complement-free
ones, let $\os^*$ be a
social welfare-maximizing outcome of the composition of auctions for
valuations~$\vals$.  For each player~$i$, with valuation $v_i =
(v^1_{i},\ldots,v^m_{i})$ with each $u_i(\os; v_i^{\ell})$ additive,
define the {\em proxy valuation} $v^*_i$ as a valuation~$v^{\ell}_{i}$
that satisfies $u_i(\os^*; v_i^{\ell}) = u_i(\os^*; v_i)$.
With the valuation profile $\vals^*$, all players have additive
utility functions.  By the previous paragraph, we can define action
distributions~$D^*_1(\vals^*),\ldots,D^*_m(\vals^*)$ such that, for
every action profile $\acs \in \A$,
\begin{equation}\label{eq:im}
\sum_i \E_{a_{i}^*\sim D_{i}^*(\vals^*)}\left[u_{i}(a_{i}^*,\acs_{-i};
 v^*_{i})\right] \geq
\lambda \opt(\vals^*) - \mu \rev(\acs).
\end{equation}

The utilities on the left-hand side and social welfare on the
right-hand side of~\eqref{eq:im} refer to the proxy valuations
$\vals^*$, not the true valuations $\vals$.
The complement-free condition~\eqref{eq:cfutils}
and choice of~$v_i^*$ imply that
$u_i(\os^*; v_i) = u_i(\os^*; v^*_i)$ and
$u_i(\os; v_i) \ge u_i(\os; v^*_i)$ for every player~$i$ and outcome
$\os \in \Ou$.
Hence, the left-hand side of~\eqref{eq:im} is at most
$\sum_i \E_{a_{i}^*\sim D_{i}^*(\vals^*)}[u_{i}(a_{i}^*,\acs_{-i};
 v_{i})]$.
Also, since $\os^*$ is a feasible outcome, the optimal welfare with
the proxy valuations~$\vals^*$ is at least as large as
with the true valuations~$\vals$:
\begin{equation*}
\opt(\vals^*) \geq \sum_i u_i(\os^*;v_i^*)+R(\os^*) = \sum_i u_i(\os^*;
\val_i)+R(\os^*)= \opt(\vals),
\end{equation*}
and so the right-hand side of~\eqref{eq:im} is at least $\lambda \opt(\vals)
- \mu \rev(\acs)$ (for every $\acs$).
Combining~\eqref{eq:im} with these two inequalities verifies the
smoothness condition for the auction composition.
\end{proof}

\paragraph{A Promise Fulfilled.}
In the Introduction, we mentioned that the following result follows
easily from the machinery developed in this survey.
\begin{theorem}\label{t:sm}
Every Bayes-Nash equilibrium of the simultaneous first-price auction
game with bidders with independent submodular valuations
and quasi-linear utilities achieves expected social welfare at least
$1-1/e$ times the expected optimal welfare.
\end{theorem}
To review, Theorem~\ref{t:sm} follows immediately from chaining
together the following tools: the fact that the first-price
single-item auction is $(1-\tfrac{1}{e},1)$-smooth
(Theorem~\ref{thm:single-item-fpa}); the fact that submodular
functions are complement-free and hence induce complement-free
utilities for players with quasi-linear utility functions; the
composition theorem, which implies that simultaneous
first-price auctions are $(1-\tfrac{1}{e},1)$-smooth for such players
(Theorem~\ref{thm:composition}); and the extension theorem stating
that every Bayes-Nash equilibrium of a $(\lambda,\mu)$-smooth
auction with independent player valuations has expected social welfare
at least $\tfrac{\lambda}{\mu}$ times the maximum possible
(Theorem~\ref{thm:extension-incomplete}).

Of course, this same machinery applies far more generally ---
whenever an auction can be expressed as the simultaneous composition
of smooth auctions and players have complement-free utility functions
over the outcomes of these auctions.

\paragraph{Relationship to Theorem~\ref{thm:sim-fpa-bne}.}
To see that the proof of Theorem~\ref{thm:composition} generalizes that of
Theorem~\ref{thm:sim-fpa-bne}, suppose that each of the
$m$ auctions is a first-price single-item auction and that each player
has a unit-demand valuation.
Recall that such a valuation~$v_i=(v_{i1},\ldots,v_{im})$ can be
written as the maximum of $m$ additive valuations
$\a^1_i,\ldots,\a^m_i$, where $w^{\ell}_{ij}$ equals $v_{ij}$ if $\ell
= j$ and 0 otherwise.
The proof of Theorem~\ref{thm:composition} defines proxy additive
valuations $\vals^*$ by
setting each $\val_i^*$ to $\a^{j^*(i)}_i$, where $j^*(i)$ is the item
$i$ receives in a welfare-maximizing allocation (or to the all-zero
function, if $j^*(i)$ is not defined).  The deviation $D^*_i(\vals^*)$
of player~$i$ is then defined by deviating independently in each of
single-item auctions~$j$, as if its valuation for that item is
$w^{j^*(i)}_{ij}$.  The deviation used to prove smoothness of a
first-price single-item auction is to bid half of one's valuation
(Theorem~\ref{thm:single-item-fpa}), so player~$i$'s deviation here is
to bid $v_{ij^*(i)}/2$ on item $j^*(i)$ and~0 on every other item.
These are exactly the deviations used in the proof of
Theorem~\ref{thm:sim-fpa-bne}.

\paragraph{From simpler to more complex valuations.}
The heart of the proof of Theorem \ref{thm:composition} extends a
smoothness inequality from additive player utility functions to
complement-free utility functions.  The same argument shows more
generally that whenever a mechanism is smooth for some class of
utilities and valuations, it is equally smooth when the utility
function of each player is a maximum over such utilities and
valuations.
\citeA{Feldman2015} and \citeA{Lucier2015} give applications of this
generalization.

\trdelete{
Thus we simply need to show that for any monotone, submodular valuation $v_i$ and for any allocated set $S$, we can write:
\begin{equation}
\etedit{v_i(S) = \max_{\ell \in \El_i} \sum_{j\in S} w_{ij}^{\ell}}
\end{equation}
\etcomment{I really disliked the notation that was here, so propose to change it. Old notation is in the tex file commented out.}
i.e. that any submodular valuation can written as a maximum over a set of additive valuations. The latter property is known in the literature as the fact that submodular valuations are a subset of XOS valuations, the latter being all valuations that can be expressed as a maximization over additive valuations (see e.g. \cite{Lehmann2001}).
}
\trdelete{
The proof of this fact is easy and neat and we present it here for completeness: We will create an proxy additive valuation $w_i^S$ for each set $S$. Then consider the items in $S$ sequentially and in some arbitrary order and set $w_{ij}^S$ to be the marginal value of acquiring item $j$ when the player already has all items of $S$ that were ordered before $j$. More formally, if $S_{<j}$ are all the items of $S$ ordered before $j$, then $w_{ij}^S = v_i(\{j\}\cup S_{<j}) - v_i(S_{<j})$. Moreover, $w_{ij}^S=0$ for any $j\notin S$. \etedit{By construction, $v_i(S) = \sum_{j\in S} w_{ij}^S$, and it doesn't depend on the order.}
}

\trdelete{
 \etedit{Now we need to show that for any other set $S'$ we have the inequality} $v_i(S') \geq \sum_{j\in S'}w_{ij}^{S}$. Since $w_{ij}^S=0$ for $j\notin S$ and since valuations are monotone, it suffices to show the latter property for any $S'\subseteq S$. If one considers the items in $S'$ respecting the same order as the one used for constructing $w_i^S$, then by submodularity and since $S'_{<j}\subseteq S_{<j}$, the marginal value when adding item $j$ to $S'_{<j}$ is at least as large as when adding it to $S_{<j}$, which is equal to $w_{ij}^S$. Since $v_i(S')$ is equal to the sum of these marginal values, we get: $v_i(S') \geq \sum_{j\in S'} w_{ij}^S$. These two facts allows us to write for any set $S'$: $v_i(S') = \max_{S} \sum_{j\in S'} w_{ij}^{S}$, which is what we wanted.
\end{example}
}


\section{Impossibility Results}
\label{sec:lower_bounds}

Can we improve over Theorem~\ref{t:sm}, either in the approximation
guarantee or in the generality of players' valuations?  If not with
simultaneous first-price auctions, then what about with some other
``simple'' auction format?  This section outlines a general technique
for ruling out good price-of-anarchy bounds for simple auctions.

Consider a setting with $n$ bidders and $m$ items.  We now impose no
restrictions on the valuation~$v_i$ of each bidder~$i$, other than
monotonicity (meaning $v_i(S) \le v_i(T)$ whenever $S \subseteq T$) and
normalization (meaning $v_i(\emptyset)=0$).  Such general valuations
permit complementarities between items, where one item is valuable
only in the presence of other items (e.g., the left and right shoes
of a pair).  An extreme example is a {\em single-minded}
valuation~$v_i$, which is defined by a subset $S$ of items and
nonnegative number $w$ as
\[
v_i(T) = 
\left\{
\begin{array}{cl}
w
&
\text{if $T \supseteq S$;}\\
0
&
\text{otherwise.}
\end{array}
\right.
\]

Can we extend the guarantee in Theorem~\ref{t:sm} for simultaneous
first-price auctions to general valuations, or at least to
single-minded bidders?
\citeA{Hassidim2011} provided a negative answer,
even for the complete-information case.

\begin{theorem}[\citeA{Hassidim2011}]\label{t:hassidim}
With general bidder valuations, simultaneous first-price auctions can
have mixed Nash equilibria with expected welfare arbitrarily smaller
than the maximum possible.
\end{theorem}
For example, equilibria of simultaneous first-price auctions need not
obtain even 1\% of the maximum-possible welfare when there are
complementarities between many items.  The examples in the proof of
Theorem~\ref{t:hassidim} make use of one single-minded bidder and one
unit-demand bidder per item.

Can we overcome the negative result in Theorem~\ref{t:hassidim} by
designing a different auction?  
The rule of thumb in practice is that
simple auctions can perform poorly with complementarities between
items (see e.g.~\citeA{milgrombook}).  One reason for this is the
``exposure problem,'' where a bidder that has value only for a bundle
of items runs the risk of acquiring only a strict subset of her
desired bundle, at a significant price.
But how would we prove a rigorous version of this rule of thumb?  
For a fixed auction format, like simultaneous first-price auctions, the
obvious way to rule out good price-of-anarchy bounds is via an
explicit example (as in Theorem~\ref{t:hassidim}).  How could we rule
out good bounds for all simple auctions simultaneously?

The key idea is to proceed in two steps.  The first step is to rule out
good approximation algorithms for the computational problem of
allocating the items to maximize the social welfare.
The second step is to show that there is a simple auction
with good equilibria only if there is a good approximation
algorithm for the welfare-maximization problem.

The first step is implemented in~\citeA{DNS05}.  The formal statement is
about the {\em communication complexity} of the welfare-maximization
problem~\cite{KN96,fttcs}.  Imagine a setup where each player~$i$
initially knows only
her own valuation, and the players cooperate to (approximately)
maximize the social welfare by exchanging information (in the form of
bits) about their valuations.  A {\em communication protocol}
specifies how information is exchanged --- who tells what to whom, as
a function of the history-so-far.  The {\em cost} of a protocol is the
maximum number of bits exchanged over all possibilities for the
players' valuations.  The communication complexity of a problem is the
minimum cost of a communication protocol for it.
\begin{theorem}[\citeA{DNS05}]\label{t:dns}
For every constant $\alpha > 0$, every communication protocol that
achieves an $\alpha$-approximation of the maximum social welfare for
general player valuations has cost exponential in the number of
items~$m$.
\end{theorem}
Theorem~\ref{t:dns} is proved via a reduction from a canonical hard
communication problem (a version of the ``Disjointness'' problem);
further details are outside the scope of this survey.

The second step is to extend the lower bound in Theorem~\ref{t:dns}
from communication protocols to equilibria of ``simple'' auctions.
For the following result, by ``simple'' we mean that the number of
bids submitted by each player is subexponential in the number of
items~$m$.  For example, simultaneous single-item auctions require
only a linear number of bids per player.  A ``direct-revelation''
mechanism, where each player submits her type, requires $2^m$ bids
per player (one
for each subset of items) and thus does not qualify as simple.
\begin{theorem}[\citeA{Roughgarden2014}]\label{t:lb}
If every $\alpha$-approximate communication protocol for the
welfare-maximization problem has cost exponential in the number of
items~$m$, 
then no family of simple auctions guarantees equilibrium welfare at
least $\alpha$ of the maximum possible.
\end{theorem}
The proof of Theorem~\ref{t:lb} effectively shows how to compute an
equilibrium allocation of a simple auction without using too much
communication.  In the presence of a good price-of-anarchy bound,
this constitutes a low-cost communication protocol with a good
approximation guarantee.  If the latter cannot exist, then neither can
the former.

Technically, Theorem~\ref{t:lb} (and Corollary~\ref{cor:lb} below) are
only known to hold for
$\epsilon$-approximate equilibria --- meaning every player mixes only
over strategies with expected utility within $\epsilon$ of a best
response --- where $\epsilon > 0$ can be taken arbitrarily small.  
Theorem~\ref{t:lb} applies even to (approximate) mixed Nash equilibria
in the complete-information setting.

Combining Theorems~\ref{t:dns} and~\ref{t:lb} makes precise the
empirical rule of thumb that simple auctions cannot guarantee high
welfare outcomes when there are complementarities between items.
\begin{corollary}\label{cor:lb}
With general bidder valuations, no family
of simple mechanisms guarantees equilibrium welfare at least a
constant fraction of the maximum possible.
\end{corollary}

There are analogous results for restricted classes of valuations.
For example, when players' valuations are subadditive (meaning $v_i(S
\cup T) \le v_i(S) + v_i(T)$ for all $S,T$), combining another result
in~\citeA{DNS05} with Theorem~\ref{t:lb} shows that no family of
simple auctions guarantees equilibrium welfare more than 50\% of the
maximum possible.  Since simultaneous first-price auctions do guarantee
at least 50\% of the maximum-possible welfare at equilibrium
\cite{Feldman2013}, they are optimal simple auctions for subadditive
player valuations in a precise sense.  It is an open question to
determine the best equilibrium welfare guarantee achievable by simple
auctions for submodular player valuations (see Section~\ref{sec:open}).


\section{Further Topics and a Guide to the Literature}
\label{sec:other-topics}

This section surveys the many related topics not covered by this
survey, with pointers for further reading.

\paragraph{Second-price-type auctions.} Throughout the survey we
focused on first-price auctions, with a cameo from all-pay auctions.
Some auctions used in practice have a different payment
scheme. Examples include the generalized second price auction used in
online ad auctions, and the uniform-price auction
used in financial institutions (e.g., to sell treasury bonds).
In these auctions, a player pays
the minimum bid she would need to make to continue to win the same
item(s).

The smoothness framework extends to such ``second-price-type''
auctions as well,
with small technical modifications to the smoothness
definition.
The primary change needed in Definition~\ref{def:smooth-auction} is
that, in~\eqref{eqn:general-smoothness}, the revenue of the bid
profile $\bids$ is replaced by the sum of the winning bids in $\bids$.
To obtain an approximate welfare guarantee, one needs to assume
that players do not bid above their valuations at
equilibrium.
This assumption implies that the sum of the winning bids
is bounded above by the
current welfare at equilibrium, and the resulting price-or-anarchy
bound is $\lambda/(1+\mu)$ (instead of $\lambda/\mu$).
Example applications
include price-of-anarchy bounds for
simultaneous second-price auctions
\cite{Christodoulou2008,Bhawalkar2011}, the generalized second-price
auction \cite{PaesLeme2010}, and uniform-price auctions
\cite{Markakis2012}.
\citeA{Roughgarden2012} and \citeA{Syrgkanis2013}
explain how to interpret the analyses in these papers as smoothness proofs.

One major technical issue in second-price-type auctions is the
no-overbidding assumption. Even in a single-item second-price auction,
where bidding truthfully is a weakly dominant strategy
and leads to an efficient allocation, there are very inefficient equilibria
if people overbid.\footnote{Think of two players, with player $a$
having valuation $\epsilon$ and player $b$ having valuation $1$. Player $a$
bidding $1$ and the player $b$ bidding $0$ is an equilibrium
of the auction with welfare only $\epsilon$ times the maximum possible.}
Thus the no-overbidding assumption is essential for equilibrium
efficiency guarantees.
In some cases, such as the
generalized second-price auction and the uniform-price auction,
the no-overbidding assumption can be justified by proving that
overbidding is a weakly dominated strategy.
In more complex settings, such as simultaneous second-price
auctions, overbidding may be unnatural but is not
always weakly dominated.
See
\citeA{Christodoulou2008,Bhawalkar2011,Roughgarden2012,Syrgkanis2013,Feldman2013}
for further discussion.

\paragraph{Further applications of the smoothness framework.} There are several
auction formats and mechanism design settings where the smoothness
framework has been employed to characterize the price of anarchy, or
where in retrospect existing price-of-anarchy analyses can be cast as a
smoothness proof. \teedit{We have seen} 
\emph{first-price
auctions} 
\cite{Hassidim2011,Syrgkanis2013} and \emph{all-pay auctions}
\cite{Syrgkanis2013,Christodoulou2015b}.
Other auction formats include: variants of \emph{uniform-price
auctions} that are used widely in financial
markets \cite{Markakis2012,Markakis2013}, combinatorial auctions where
the allocation of items is decided by a \emph{greedy algorithm}
\cite{Lucier2010}, \emph{position auctions} such as the generalized
second-price auction and generalizations of it
\cite{Caragiannis2014,PaesLeme2010,Syrgkanis2013}, the
\emph{proportional mechanism for bandwidth allocation}
\cite{Johari2004,Syrgkanis2013,Caragiannis2014b,Christodoulou2015},
the \emph{Walrasian mechanism} for combinatorial markets,
\cite{Babaioff2014},
{\em relax-and-round mechanisms}
\cite{Duetting2015}, and auctions for \emph{renewable energy markets}
\cite{Kesselheim2015}.

The price-of-anarchy bounds in all of these papers can be
interpreted as showing that the auction is smooth
(Definition~\ref{def:smooth-auction}) or, in some cases,
even smooth with private deviations
(Definition~\ref{def:semi-smooth-auction}).
As a rule of thumb,
when the auction allows players to express their entire private
valuation through the action space, then the auction is likely
to satisfy smoothness with private deviations
(e.g., the generalized second-price auction, relax-and-round mechanisms,
and the Walrasian mechanism).
If the action space is
restricted and only allows the player to bid some proxy restricted
valuation, then the auction tends to only satisfy the weaker
smoothness condition that extends only to the independent private
values setting (e.g., simultaneous single-item auctions).
The key intuition is that in the latter auction formats, most
price-of-anarchy proofs
require some version of the argument in the proof of
Theorem~\ref{thm:composition}
that extends a smoothness equality from
some restricted class of valuations
to some more general class of valuations.
However, this extension only
establishes the weaker smoothness condition for general valuations, even if
the auction is smooth with private deviations for the restricted
class.

\paragraph{Price-of-anarchy bounds via non-smooth techniques.}
Not all of the known price-of-anarchy bounds for auctions are
smoothness proofs.
One nice example is due to \citeA{Feldman2013}, who analyze
simultaneous first- and second-price item auctions
when players' valuations are subadditive (meaning $v_i(S \cup T) \le
v_i(S) + v_i(T)$ for all subsets $S,T$ of items).  Subadditive
valuations are strictly more general than the completment-free
valuations studied in Section~\ref{sec:composability}.
\citeA{Feldman2013} prove that the price of anarchy is constant in
this case, while the
smoothness framework is only known to
give a bound that degrades logarithmically with the number of items
\cite{Bhawalkar2011,Hassidim2011,Syrgkanis2013}.
\citeA{Feldman2013} deviate from the smoothness framework by using a
different set of deviating actions for each Bayes-Nash equilibrium.
\citeA{Duetting2013} propose a
generalization of the smoothness framework that includes even the
analysis of \citeA{Feldman2013} as a special case.

A new approach to bounding the price of anarchy
was recently proposed by \citeA{Kulkarni2015}.
For an arbitrary game, they formulate a convex program that has a
strong connection
with the welfare-maximization problem associated with the
game.
They prove that every vanishing regret sequence of action profiles
(Definition~\ref{d:vanishing})
of the game
can be associated with a solution to the Fenchel dual of
this convex program, with the average welfare of the sequence being close
to the value of the dual at this solution.
Using duality, one can argue that the average welfare of every such
sequence is close to the optimal welfare.
The primary applications in \citeA{Kulkarni2015} concern routing,
scheduling, and location games.
Interesting open questions include whether or not this framework can
prove better bounds than the smoothness framework for natural auction
games, and whether or not it extends to incomplete-information
settings.

Another line of results that fall outside of the smoothness framework
are
those that prove stronger bounds for restricted subclasses of
equilibria.
For example,
\citeA{Bikhchandani1999} proves that pure Nash
equilibria of simultaneous first-price auctions are fully efficient
(Theorem~\ref{thm:sim-fpa-pne}).  Since mixed-strategy Nash equilibria
need not be fully efficient (Example~\ref{ex:badmne}), this result
cannot be established by a smoothness proof.
For another example, \citeA{Bhawalkar2011} show that, assuming
subadditive player valuations and no overbidding, every pure Nash
equilibrium of simultaneous second-price auctions has welfare at least
50\% of the maximum possible.  Again, this guarantee does not hold for
mixed Nash equilibria and hence cannot be established via a smoothness
proof.
Further examples are provided by
\citeA{Christodoulou2015,Christodoulou2015b}, who prove
better bounds for mixed Nash equilibria of complete-information
first-price and all-pay auctions than are known for Bayes-Nash
equilibria in the more general incomplete-information case.

\paragraph{Sequential auctions.}
In all of the auctions considered in this survey,
actions are chosen simultaneously by the players and the allocation
and payments are decided in one shot.
Many auctions in practice have a
sequential component to them.
The simplest-possible theoretical model
is that of sequential item auctions, where items are sold one-by-one
in some predefined order
via single-item auctions.
This model has a long history in economics, starting from the work of
\citeA{Milgrom1982,Milgrom1999}. Sequentiality leads to a host of new
complications in price-of-anarchy analyses.
The biggest problem is
that when a player deviates at some stage of the game, the deviation
can cause a ripple effect in subsequent stages, changing for instance the
prices of subsequent items. To bypass this complication,
\citeA{PaesLeme2012} and \citeA{Syrgkanis2012a} propose a bluffing
deviation, where
a player pretends to play as in equilibrium of some random type, until
the right moment arrives when she deviates to acquire some item. Such a
deviation analysis, where the deviation involves simulating some
current equilibrium behavior, extends to the incomplete-information
setting in a black-box manner
\cite{Syrgkanis2013}. Technically, one can show that most of the
smoothness theory extends even if the deviations $a_i^*$ can depend on
the current action $a_i$ from which the player is deviating from. Such
deviations unlock the ability to analyze sequential auctions from a
price-of-anarchy point of view.

\paragraph{Budget constraints.} Most of price-of-anarchy analyses
of auctions
assume that players have quasi-linear preferences, meaning that
a player's utility
is her valuation for the items received less the payment made.
This implicitly assumes that players are capable of paying an
arbitrarily large amount.
In many practical auction settings,
players also have budget constraints.
The simplest way to model a budget is to define the utility of a
player as $-\infty$ if her payment exceeds her budget.
Do the efficiency guarantees for smooth auctions extend to the setting
where players' have finite budgets?

High prices can prevent
a player with a small budget and a high valuation from obtaining the
allocation she would get in a welfare-maximizing outcome.
To address this issue, a sequence of papers have used an alternative
benchmark to measure the efficiency of auctions in the presence of
budgets.
This benchmark is called
optimal effective welfare \cite{Syrgkanis2013} or the
optimal liquid welfare \cite{Dobzinski2014},
and is defined as the
maximum welfare achievable after capping each player's valuation for an
allocation by her budget. \citeA{Syrgkanis2013} show that, under minor
additional assumptions, the
equilibrium welfare of a $(\lambda,\mu)$-smooth auction is at
least $\frac{\lambda}{\mu}$ times the expected maximum effective
welfare. This benchmark was further analyzed by
\citeA{Caragiannis2014b} for the proportional bandwidth allocation
mechanism.
One interesting open question is
whether one can also lower bound the effective welfare at equilibrium,
rather than just the actual welfare. \citeA{Caragiannis2014b} provide such
a result for the proportional bandwidth allocation mechanism.

\paragraph{The price of anarchy for revenue.}
All of the price-of-anarchy guarantees covered in this survey concern
the equilibrium welfare of auctions.
Another quantity of primary importance is the revenue
of the auctioneer.
It is analytically intractable to explicitly characterize
the expected revenue of most non-truthful auctions at equilibrium,
even in simple settings such as asymmetric single-item first-price
auctions (see e.g.~\citeA{Kirkegaard2014}).

Remarkably, \citeA{Hartline2014} show that in ``single-parameter''
settings, where the
private information of a player boils down to a single number,
a variant of the smoothness framework can be used to prove lower
bounds on the expected auction revenue at a Bayes-Nash equilbrium,
relative to the maximum-possible expected revenue.
They formulate a stronger
version of smoothness, asserting a smoothness-type inequality
on every player individually (rather than in
aggregate).  They show that this
stronger condition implies approximate revenue guarantees in several
settings, for example when there is no bidder with a unique valuation
distribution, and
when the auction is augmented with the appropriate reserve price.
This implies, for example, that the expected revenue at a
Bayes-Nash equilibrium of an
asymmetric first-price single-item auction with an appropriate reserve
price is at least $\frac{1}{4}$ times the maximum expected revenue.
\trdelete{An
important open direction is whether one can make similar statements in
``multi-parameter'' settings,
where a player's private information cannot be captured by a single
number (like in multi-item auctions).}

\paragraph{Algorithmic characterizations of smoothness.}
Smooth auctions have many advantages, including composition theorems,
and extension theorems for incomplete-information and no-regret
learning settings.
Which allocation algorithms and payment rules yield smooth auctions?


Several papers provide sufficient (but not necessary) algorithmic
conditions for smoothness. \citeA{Lucier2010} show that if the
allocation is based on a large class of greedy $c$-approximation
algorithms with a ``loser-independence'' property, then coupling
it with a first- or second-price payment scheme leads to a smooth
auction with price of anarchy $\Theta(c)$.
\citeA{Lucier2015} gave a more general characterization along these
lines: if an auction can be viewed as
running a greedy algorithm in
some abstract ``element space,'' subject to matroid or poly-matroid
constraints, then coupling it with a first-price payment scheme
leads to a smooth auction.
Examples of auctions that fall into this characterization are
simultaneous item auctions, uniform-price
auctions and position auctions, even in the presence of externalities
\cite{Roughgarden2012b}.
\citeA{Babaioff2014} prove
that in a combinatorial auction setting with gross substitute
valuations, maximizing welfare with respect to the reported valuations
and charging suitable payments yields a smooth auction.
\citeA{Duetting2015} show that if an
allocation algorithm is based on a technique in approximation
algorithms known as ``relax-and-round,'' then coupling it with
suitable payments yields a smooth auction.

\citeA{Duetting2015b} give a necessary and sufficient condition for
smoothness in single-parameter settings.
This characterization enables impossibility results for smooth
mechanisms that are independent of any computational concerns.
\citeA{Duetting2015b} define
the \emph{permeability} of an allocation algorithm and show
that, for every single-parameter setting, there exists a smooth
auction with approximation guarantee~$c$ if and only if
there exists an allocation
algorithm for the setting with permeability $\Theta(c)$.
Variants of permeability were used earlier as a sufficient condition
for smoothness in
\citeA{Syrgkanis2013} and
\citeA{Hartline2014}.
\citeA{Duetting2015} show how their characterization yields
impossibility results,
for example for combinatorial auctions with single-minded
bidders.
\trdelete{Extending this characterization to
multi-dimensional private information settings is a very interesting
open direction.}

\paragraph{Valuations with complementarities.} Most of the results presented in
this survey assume that the valuation function of each
player exhibits no complementarities.
The impossibility results in Section~\ref{sec:lower_bounds} show that
some assumption of this type is necessary for good price-of-anarchy
bounds.
A natural goal is to prove welfare guarantees for auctions that
degrade gracefully with the ``degree of complementarity'' of players'
valuations.
Such an analysis was done for computationally efficient
truthful mechanisms in \citeA{Abraham2012},
and for the equilibria of non-truthful auctions in
\citeA{Feldman2015}.
The results in \citeA{Feldman2015} include, for example,
an extension of the composition theorem in Section
\ref{sec:composability} (Theorem~\ref{thm:composition}) beyond
complement-free valuations.

\paragraph{The price of anarchy in large markets.}
A classic economic
intuition is that as a market grows large, the effect of each
player on the market and hence the opportunities for strategic
behavior
should diminish,
leading to more efficient outcomes.
Making this idea rigorous is non-trivial.
\citeA{Swinkels2001} showed that in a uniform-price auction,
where there is generally inefficiency in small markets,
the equilibrium welfare converges to the
optimal welfare as the number of players and units of the good go to
infinity, under a noise assumption on the arrival of players or
units. \citeA{Feldman2015b} extended these results to
combinatorial markets with arbitrary bidder valuations across
different goods and to the simultaneous uniform-price auction,
offering an adaptation of the smoothness framework that can incorporate
large market assumptions. \citeA{Cole2015} subsequently
provided similar full-efficiency, large market results for the
``Walrasian mechanism'' with gross substitute
valuations.
One benefit of using a smoothness approach to prove such large market
results is that it can prove full efficiency guarantees even when
players do not behave straightforwardly in the limit (e.g., when
players cannot fully express their valuations through their bids).


\paragraph{Complexity of computing an equilibrium.} Another line of
work addresses equilibrium computation in simple
auctions. \citeA{Cai2014} analyze the complexity of
computing a Bayes-Nash equilibrium in simultaneous second-price
auctions and shows that it is a computationally hard problem.
Such computational intractability raises questions about the
predictive power of the equilibrium concept.

When a Nash or Bayes-Nash equilibrium is hard to compute, no-regret
learning often comes to the rescue.
When the strategy space of each player is part of the input,
\teedit{the players can use no-regret algorithms (as discussed in Section~\ref{sec:no-regret}) to approximately compute learning outcomes, also called ``coarse correlated equilibrium'' (which is closely related to
Definition~\ref{d:vanishing}) in polynomial time. The }
price-of-anarchy guarantees proved using smoothness arguments apply to such equilibria (Section~\ref{sec:no-regret}).

With implicitly defined exponential-size strategy spaces,
as in simultaneous item auctions with many items, it is an
open question whether polynomial-time decentralized dynamics exist
that lead to no-regret outcomes. On the negative side,
\citeA{Daskalakis2016} show that, under
appropriate complexity assumptions, there are no
polynomial-time algorithms that guarantee no regret in the worst case.
Conceivably, one could dodge this hardness result by
designing coordinated dynamics for all of the players (to avoid the
hard instances of the learning problem).
\citeA{Daskalakis2016} complement their
impossibility result by proposing a relaxed version of no-regret
learning, called ``no-envy learning.''  They show that there exist
polynomial-time no-envy learning algorithms,
and that most
price-of-anarchy bounds proved via smoothness arguments extend to
outcome sequences generated by no-envy learners.

Motivated by equilibrium tractability,
\citeA{Devanur2015} seek combinatorial auctions that simultaneously
have good price-of-anarchy bounds and also a very small strategy
space. They propose an algorithm where
the strategy space of each player is a single number, with equilibrium
welfare at least an inverse logatihmic (in the number of items)
fraction of the maximum possible.
\citeA{Braverman2016} show that a large class of auctions cannot
improve over this logarithmic bound.

Finally, \citeA{Christodoulou2008} and \citeA{Dobzinski2015}
consider the complexity of computing a complete-information pure Nash
equilibrium of simultaneous second-price auctions.
In simultaneous first-price auctions, such equilibria correspond to
Walrasian equilibria \cite{Bikhchandani1999}, and can therefore be
computed using linear programming (whenever one exists).


\paragraph{Bayesian no-regret learning and Bayes-coarse correlated
equilibria.} Section \ref{sec:no-regret} analyzed no-regret
learning when players' valuations remain fixed over time. Recent
work of \citeA{Hartline2015} shows that the efficiency guarantees of a
smooth auction hold even when each player's valuation is drawn anew at
each iteration,
independently from some distribution $\F_i$.
The first step of the argument shows that if all players use no-regret learning
algorithms in such an environment, then the empirical distribution of
joint play converges to an analog of the
coarse correlated equilibrium for games of incomplete
information (``strategic form Bayes-coarse
correlated equilibrium'').
The second step shows that the approximate efficiency guarantees
of smooth auctions apply even to such equilibria.
\citeA{Caragiannis2014} had previously showed
such a result
only under the stronger condition of smoothness via private
deviations (Definition~\ref{def:semi-smooth-auction}).

\paragraph{Signaling and Bayes-correlated equilibria.}
This survey discussed only private value settings where each player
knows her own valuation and nothing else.
\citeA{Caragiannis2014} show that for auctions that are smooth via
private deviations, the approximate welfare guarantees hold
even if players
receive arbitrary signals about others' valuations prior to
bidding. The set of outcomes that can arise
in such a setting is closely related
to the notion of a Bayes-correlated
equilibrium
\cite{Bergemann2011,Bergemann2013}.
Specifically, this result implies that
in a $(\lambda,\mu)$-smooth auction via private
deviations, every Bayes-correlated equilibrium  where players know at
least their own valuation achieves expected welfare at least
$\frac{\lambda}{\mu}$ times the expected optimal welfare.


\section{Open Questions and Research Directions}
\label{sec:open}

We conclude with a dozen suggestions for future research.

\subsection{Stronger Price-of-Anarchy Bounds for Common Auction Formats}

This survey focused on auction formats for which the price of anarchy
is relatively well understood, but some open questions remain even for
these auctions.

\paragraph{1. The price of anarchy of first-price single-item auctions.}
What is the exact price of anarchy of asymmetric first-price
single-item auctions with independent player valuations?  The answer
is at least~$1-1/e$ (Theorem~\ref{thm:single-item-fpa}, due to
\citeA{Syrgkanis2013}) and at most~$.87$ \cite{Hartline2014}.

\paragraph{2. Better simple auctions with submodular valuations.} For
subadditive player valuations, the communication complexity lower
bounds in \citeA{Roughgarden2014} (Section~\ref{sec:lower_bounds})
imply that simultaneous first-price auctions have the best-possible
price of anarchy of any simple auction.
With the stronger assumption of submodular bidder valuations
(Section~\ref{sec:composability}), the price of anarchy of
simultaneous first-price auctions is~$1-1/e$ (Theorem~\ref{t:sm}), and
this is tight in the worst case \cite{Christodoulou2013}.
The state-of-the-art in communication complexity only implies that
there is no simple auction with price of anarchy better
than~$1-1/2e$ \cite{DV13}.  Intriguingly, there is a low-cost communication protocol
that approximates the welfare to within a factor (slightly) larger
than $1-1/e$ \cite{FV10}.  Is there an analogously good simple
auction?  What is the best price of anarchy
achievable by a simple auction with submodular player valuations?


\paragraph{3. Explicit impossibility results.}
In some non-auction settings, such as in routing and congestion games,
there is a ``generic'' construction that always produces examples with
the worst-possible price of anarchy \cite{Roughgarden2009}.
Is there a similiar generic construction for some family of auctions,
for example simultaneous first-price auctions,
perhaps generalizing the examples in \citeA{Christodoulou2015} and
\citeA{Feldman2015b}?
Impossibility results derived from communication complexity
(Section~\ref{sec:lower_bounds}) do not seem to lead to such a
construction.

\subsection{Price-of-anarchy Bounds for Other Auction Formats}

There are a number of practically relevant auction formats that have
been understudied from a price-of-anarchy perspective.

\paragraph{4. Procurement auctions.} The price-of-anarchy literature
has focused on forward auctions, where buyers with private
valuations compete for resources in an auction. In many
applications, for example the procurement of energy from power firms,
the situation is the opposite: a set of sellers with private
costs compete to provide some service to one or more buyers who are
running an auction to obtain the service as cheaply as
possible.
Simple non-truthful procurement auctions are common in practice, and
not much is known about their price of anarchy.
See \citeA{Babaioff2014b} for some initial results in this direction.

\paragraph{5. Double auctions.} Even less is known about simple double
auctions, where an intermediary runs an auction that has both buyers
and sellers with private valuations.
The classic impossibility result of \citeA{Myerson1983}
for worst-case instances suggests that assumptions are
needed for positive results. A line of work in
economics \cite{Rustichini1994,Satterthwaite2002,Cripps2006} shows
convergence to full efficiency at equilibrium of simple auctions as
the market grows large. Does the smoothness framework for large
markets proposed in \citeA{Feldman2015b} extend to large double
auctions, ideally unifying all previous such results in a single
analysis framework?

\paragraph{6. The combinatorial clock auction.}
Iterative auctions, which take place in rounds rather than in
a single shot, pose a particular challenge for price-of-anarchy analyses.
For example, the iterative combinatorial clock
auction \cite{Ausubel2006} is widely used in practice for
selling wireless spectrum. This auction contains both sequential and
simultaneous bidding elements, which complicates the analysis of
its price of anarchy.
\citeA{Bousquet2016} recently provided
a theoretical analysis of the efficiency of the auction
when all players act truthfully. Extending this approximate welfare
guaranteee to equilibria appears highly non-trivial.

\paragraph{7. Mechanisms without money.}
In many of the ``killer applications'' of mechanism design,
including kidney exchange and residency matching,
monetary transfers are prohibited.
These are examples of ``mechanism design without money''
\cite{Schummer2007,Procaccia2009}. Can we quantify the
equilibrium inefficiency of simple mechanisms
in such settings, perhaps with a smoothness-type framework?
In many of these applications, the first challenge is to identify
a well-motivated objective function for which price-of-anarchy bounds
might be possible.

\trdelete{
In such environments it is also natural to address other measures of
quality of equilibria such as fairness. Is it possible to have a
general theory of price of anarchy for fairness related metrics?
Recent positive results for truthful mechanisms \cite{Cole2013}, hint
at similar positive price of anarchy results for simple non-truthful
mechanisms.
}

\subsection{Richer Utility Models}

Most of the results in this survey concern bidders with independent
private values and with quasi-linear utility functions.  To what
extent can these assumptions be relaxed?

\paragraph{8. Restricted correlation of private values.}
Equilibria of smooth auctions can have very low welfare when players'
valuations are arbitrarily correlated
\cite{Bhawalkar2011,Feldman2013}.
Are there natural forms of valuation correlation for which
smooth auctions are still guaranteed to have near-optimal equilbiria?
For example, {\em affiliation} is a strong form of positive
correlation that has unlocked several results in economics
\cite{Milgrom1982}.
Can we bound the inefficiency of simple multi-item auctions with some
form of affiliated private valuations?

\paragraph{9. The price of anarchy with risk-averse players.}
Classical auction theory concerns risk-neutral bidders, while in
practice many bidders are risk-averse (all else being equal,
preferring low-variance outcomes to high-variance ones). Risk aversion
poses a host of
problems to the analysis of auctions, and only partial results are
known in economics on understanding equilibria in this case
(see \citeA{Maskin1984} and several follow-up works).
Can we at least bound the equilibrium efficiency of a first-price
single-item auction with risk-adverse bidders?
See \teedit{\citeA{FuHartlineHoy2013} for some results on revenue of such 
auctions and}
\citeA{Lianeas2015} for some results along these lines in
non-auction domains, such as selfish routing networks.

\subsection{Generalizing the Smoothness Framework}

The smoothness framework presented in this survey is already rather
general, but good researchers are greedy and always want more.

\paragraph{10. Beyond no-regret dynamics.} The smoothness approach
directly extends to adaptive game-playing when
the vanishing regret condition holds
(Section~\ref{sec:no-regret}).
There are a number of interesting adaptive game-playing
algorithms that are not guaranteed to achieve this condition.
%
One famous example is \emph{fictitious play}
\cite{Brown1951}, which
is only known to converge to an
equilibrium (and therefore satisfy the vanishing regret condition)
in some special
cases, such as zero-sum games \cite{Robinson1951}, potential games
\cite{Monderer1996}, and two-player games where one player has only two
strategies \cite{Berger2005}.
Does fictitious play achieve the vanishing regret
property in smooth auctions?
Is the average welfare of
fictitious play over time guaranteed to be close to the maximum
welfare in smooth auctions?
These questions are also relevant for many
other forms of adaptive game-playing;
see \citeA{Fudenberg1998} for a starting point.

\paragraph{11. Characterizations of smoothness in multi-dimensional
domains.} \citeA{Duetting2015b} provided the first algorithmic
characterization of smoothness in
single-dimensional mechanism design domains. An obvious open question
is whether or not there is an analogous
characterization for multi-dimensional
domains like multi-item auctions.
A related direction is to understand the extent to which the
approximation achievable by low-cost communication protocols
(Section~\ref{sec:lower_bounds}) characterizes the best-possible
guarantee of a simple smooth mechanism.

\trdelete{
\paragraph{Simple auctions with obvious equilibria.} Given that
Bayes-Nash equilibria can be hard to compute in simple auctions with
large strategy spaces, one can ask whether we can design simple
auctions where the equilibrium behavior is obvious and that still
maintain small price of anarchy. \cite{Devanur2015} provided auctions
where the strategy space of each player was very simple, so that
optimizing over it is easy. However, one can ask the stronger
requirement that the optimal strategy for each player is also obvious,
not just easy to compute.
}

\paragraph{12. The price of anarchy for revenue in multi-dimensional
settings.} \citeA{Hartline2014} adapted the
smoothness framework to prove bounds on the revenue of simple
auctions in single-dimensional environments. Can we obtain interesting
guarantees for the revenue of simple mechanisms for
multi-dimensional environments?
Key to the results of \citeA{Hartline2014} is the equivalence of
expected revenue and expected ``virtual welfare'' in single-parameter
settings \cite{Myerson1981}.  While this equivalence does not hold in
multi-parameter environments,
there are cases where a version of virtual welfare
well-approximates the revenue of an auction \cite{Chawla2010}.
Can we use a similar approach together with an adaptation of
smoothness to bound the revenue of other simple auctions?



\vskip 0.2in
\bibliographystyle{theapa}
\bibliography{poa_survey}

\end{document}